\documentclass[11pt]{article}

\usepackage{geometry}
\usepackage{amsmath,amssymb,amsthm,mathtools}
\usepackage{bbm}
\usepackage{microtype}
\usepackage{xcolor}
\definecolor{linkcolor}{RGB}{0,0,255}
\definecolor{citecolor}{RGB}{0,0,255}
\definecolor{urlcolor}{RGB}{0,0,255}
\usepackage[colorlinks=true,linkcolor=linkcolor,citecolor=citecolor,urlcolor=urlcolor,pdfusetitle]{hyperref}
\usepackage[noblocks]{authblk}
\usepackage{quantikz}
\usepackage{physics}
\usepackage[ruled,vlined,linesnumbered]{algorithm2e}
\usepackage{tikz}
\usepackage{enumitem}
\usepackage{graphicx}
\usepackage{subcaption}
\usepackage{xspace}

\usepackage{mleftright} 

\definecolor{dm00}{RGB}{255,255,153}
\definecolor{dm01}{RGB}{153,255,255}
\definecolor{dm10}{RGB}{153,255,153}
\definecolor{dm11}{RGB}{255,153,153}

\newtheorem{theorem}{Theorem}[section]
\newtheorem{lemma}[theorem]{Lemma}
\newtheorem{proposition}[theorem]{Proposition}
\newtheorem{corollary}[theorem]{Corollary}
\theoremstyle{definition}
\newtheorem{definition}[theorem]{Definition}
\newtheorem*{definition*}{Definition}

\newtheorem{example}[theorem]{Example}

\newcommand{\NN}{\mathbb{N}}       
\newcommand{\ZZ}{\mathbb{Z}}       
\newcommand{\RR}{\mathbb{R}}       
\newcommand{\CC}{\mathbb{C}}       
\newcommand{\EE}{\mathbb{E}}       
\newcommand{\PP}{\mathbb{P}}       
\newcommand{\DDD}{\mathcal{D}}     
\newcommand{\MMM}{\mathcal{M}}     
\newcommand{\PPP}{\mathcal{P}}     
\newcommand{\CCC}{\mathcal{C}}     
\newcommand{\OOO}{\mathcal{O}}      
\newcommand{\ii}{\mathrm{i}}       
\newcommand{\zero}{\mathbf0}       
\newcommand{\one}{\mathbf1}        
\newcommand{\EEE}{\mathcal{E}}     
\newcommand{\GGG}{\mathcal{G}}    
\newcommand{\LLL}{\mathcal{L}}     
\DeclareMathOperator{\id}{id}      
\DeclareMathOperator{\diag}{diag}  
\DeclareMathOperator{\wt}{wt}      

\newcommand{\NAME}{\mbox{SGM}\xspace} 

\title{Classical Shadows with Selective GHZ Measurements}

\author[1]{\mbox{Nicolas Faroß}}
\author[1]{\mbox{Marc Wanner}}
\author[1]{\mbox{Pingal Pratyush Nath}}
\author[2]{\mbox{Akshay Gaikwad}}
\author[2]{\mbox{Anton Frisk Kockum}}
\author[1]{\mbox{Devdatt Dubhashi}}
\affil[1]{\small Department of Computer Science and Engineering, Chalmers University of Technology and University of Gothenburg, 41296 Gothenburg, Sweden\\

\texttt{\{faross,wanner,pingal,dubhashi\}@chalmers.se}}
\affil[2]{\small Department of Microtechnology and Nanoscience, Chalmers University of Technology, 41296 Gothenburg, Sweden\\

\texttt{\{akshayga,anton.frisk.kockum\}@chalmers.se}}

\date{}

\begin{document}

\maketitle

\begin{abstract}
Classical shadows allow the prediction of many quantum-state properties from measurement data, but their efficiency strongly depends on the underlying measurement ensemble.
We introduce selective GHZ measurement (\NAME) shadows, which perform Greenberger--Horne--Zeilinger (GHZ)-basis measurements on randomly selected subsets of qubits and computational-basis measurements on their complement. The corresponding quantum circuits
are practically feasible, with their number of CNOT gates scaling linearly with the number of qubits. This is in stark contrast to global Clifford circuits, which are generally difficult to implement.
Varying the subset distribution allows tuning the \NAME shadow protocol for different classes of observables with prescribed $X/Y$ support, allowing us to outperform local Pauli as well as global Clifford shadows in certain scenarios.
We determine the exact shadow channel and its inverse for arbitrary subset distributions and prove an explicit bound on the shadow norm using methods from discrete Fourier analysis and Markov semigroup theory. For natural choices of subset distributions, this gives a polynomial scaling for observables with fixed $X/Y$-locality. Numerical simulations of up to 50 qubits support the predicted improvements for nonlocal observables and long-range systems, and demonstrate the practical benefits of our protocol's generality by showing how the subset distributions can be tailored to settings such as fermionic modes under the Jordan--Wigner transformation.
\end{abstract}

\section{Introduction}\label{sec:introduction}

Estimating properties of quantum states from measurement data is a fundamental task in quantum information. \textit{Classical shadows} provide a general framework for efficiently predicting many such properties of an unknown quantum state using relatively few measurements~\cite{huang2020predicting,huang2022learning}.
A shadow protocol applies randomly sampled unitaries from a fixed measurement ensemble to copies of a quantum state $\rho$, followed by measurements in the computational basis. Using the post-measurement ``snapshots'', one creates a classical reconstruction of the quantum state, which can then be used to estimate many observables without repeating the experiment for each.
The cost of this strategy is determined by the measurement ensemble, which sets an upper bound on the number of shots required to achieve a given precision through the shadow norm, while also determining the circuit resources required per shot. Two ensembles sit at opposite ends of this tradeoff. Local Pauli measurements are shallow and
hardware-friendly, but their sample complexity grows exponentially with the support of the target observable. At the other extreme, global Clifford measurements are essentially support-independent and highly efficient for low-rank observables, but their implementation remains challenging for many near-term devices because of their depth, entangling-gate count, and connectivity requirements.

A central open question is
\begin{quote}
    \emph{Can we find a better solution to this tradeoff between sample and circuit complexity with a measurement ensemble that is more powerful than the Pauli ensemble, but still is practical and does not require the cost of the Clifford ensemble?}
\end{quote}
This question has motivated intermediate measurement ensembles, from
shallow local Clifford circuits~\cite{bertoni2024shallow,akhtar2023scalable}
to ensembles generated by native Hamiltonian
dynamics~\cite{hu2022hamiltonian,hu2023classical,liu2024predicting}.
For many such ensembles, however, explicit expressions for the shadow
channel and its inverse remain difficult to obtain and reconstruction
often relies on numerical methods. This limits analytic guidance for tailoring
ensembles to target observables.

Given a practical measurement ensemble, the next central challenge is
\begin{quote}
    \emph{Can we rigorously characterize the shadow channel and its inverse, and obtain bounds on the shadow norm, which in turn determines the sample complexity?}
\end{quote}
These challenges have, in turn,
motivated efforts to exploit the mathematical structure of the shadow
channel, using symmetry-based
decompositions~\cite{chen2021robust,zhao2021fermionic,wan2023matchgate,bu2024classical,zhao2024group-theoretic,chang2026classical}
and frame-theoretic duality~\cite{innocenti2023shadow} to characterize
reconstruction and its statistical properties
(see Sec.~\ref{sec:related-work}).

Finally, to the best of the authors' knowledge, existing single-copy shadow protocols with sample-efficient guarantees apply to observable families whose relevant structure grows only polynomially with the number of qubits $n$. For instance, although bounds for $k$-local Pauli observables may formally cover exponentially many operators, only $\OOO(n^k)$ satisfy the locality constraint. Likewise, low-rank observables for the Clifford ensemble and constant-degree Majorana monomials for fermionic Gaussian unitaries~\cite{zhao2021fermionic} form families of size polynomial in $n$. Protocols with efficient guarantees for exponentially large observable families instead rely on simultaneous measurements of multiple copies of the state~\cite{aaronson2018shadow, king2025triply}.
This state of affairs motivates the question
\begin{quote}
\emph{Can a single-copy measurement ensemble efficiently predict an exponentially large family of observables?}
\end{quote}

We contribute to this program by introducing \emph{selective GHZ measurement (\NAME) shadows}, a family of randomized Greenberger--Horne--Zeilinger~\cite{greenberger1989going} (GHZ)-type measurement ensembles parameterized by a probability distribution over subsets of qubits. These measurements admit a shallow circuit implementation whose gate count, in particular the number of CNOT gates, scales linearly with the number of qubits in the worst case, making it well suited to near-term devices. At the same time, varying the underlying distribution allows the measurement statistics to interpolate between Pauli-like and Clifford-like regimes. We derive an analytic expression for the shadow channel eigenvalues in terms of Fourier coefficients of the underlying distribution, and prove explicit bounds on the shadow norm for structured observables using methods from Markov semigroup theory. We show that the \NAME shadow norm, and therefore the corresponding sample complexity, remains polynomial for a broad class of observables containing exponentially many Pauli operators. Our numerical results support this scaling and illustrate the flexibility of the protocol. In particular, for the fermionic setting of~\cite{zhao2021fermionic}, \NAME retains polynomial sample complexity while substantially reducing the circuit complexity.

\subsection{Main results}\label{sec:main-results}

The unitary ensembles underlying our \NAME shadow protocol are inspired by the measurement settings used for selective quantum state tomography in~\cite{patel2026selective}, where subsets of $n$ qubits are measured in one of the GHZ bases
\begin{equation}\label{eq:GHZ-type}
    \GGG_{k,m} = \Bigl\{\frac{1}{\sqrt{2}} \ket{0, x} \pm \frac{1}{\sqrt{2}}  \, \ii^k \ket{1, \overline{x}} : x \in \{0, 1\}^{m-1} \Bigr\},
    \quad
    k \in \{0, 1\}, \,
    1 \leq m \leq n ,
\end{equation}
while the remaining qubits are measured in the computational basis. In contrast to existing approaches, we do not sample these subsets uniformly, but allow for arbitrary probability distributions, which makes it possible to adjust the protocol for different classes of structured observables. For each choice of probability distribution, the corresponding \NAME unitary ensemble is defined as follows.

\begin{definition*}[Def.~\ref{def:main-ensemble}]
Let $p \colon \{0,1\}^n \to [0, 1]$ be a probability distribution. Then the $n$-qubit \emph{selective GHZ measurement (\NAME) ensemble} with distribution $p$ is given by
\begin{equation}
    \bigl\{ U_{k,s} : k \in \{0, 1\}, s \in \{0, 1\}^n \bigr\},
\end{equation}
where the unitary $U_{k, s}$ is chosen with probability $\frac{1}{2} \, p(s)$ and transforms all qubits $i$ with $s_i = 1$ from the corresponding $\GGG_{k,m}$ basis into the computational basis.
\end{definition*}

Each unitary $U_{k,s}$ can be implemented by a logarithmic-depth quantum circuit consisting of a single Hadamard gate, an optional phase gate, and a number of CNOT gates that is linear in the selected subset size (Sec.~\ref{sec:ensemble}).
This simple circuit structure can substantially reduce gate and compilation overhead on near-term quantum hardware compared to generic global Clifford measurements.

Since the \NAME ensemble is implemented by Clifford unitaries, it follows that the corresponding shadow channel $\MMM$ is diagonal in the Pauli basis (Sec.~\ref{sec:clifford-ensembles}).
Denote by
\begin{equation}
    P_{a,b} = \bigotimes_{k=1}^n \ii^{a_k b_k} X^{a_k} Z^{b_k}
\end{equation}
the standard Pauli operator with $X$- and $Z$-support encoded by $a, b \in \{0, 1\}^n$.
Then our first main result provides a simple description of the eigenvalues of the shadow channel in terms of the underlying distribution $p$ and its Fourier transform $\widehat{p}$ over $\ZZ_2^n$.

\begingroup
\let\thetheorem1
\begin{theorem}[Thm.~\ref{thm:main-channel-eigenvalues}]
Let $a, b \in \{0, 1\}^n$. Then the eigenvalue corresponding to $P_{a,b}$ for the \NAME shadow channel with distribution $p$ is given by
\begin{equation}
    \lambda_{a,b} = \begin{cases}
        \frac{1}{2} + \frac{1}{2} \widehat{p}(b), & \text{if $a = (0, \dots, 0)$}, \\
        \frac{1}{2} p\bigl(a\bigr), & \text{otherwise},
    \end{cases}
\end{equation}
where
\begin{equation}
    \widehat{p}(b) = \sum_{x \in \{0, 1\}^n} p(x) \chi_b(x),
    \qquad
    \chi_b(x) = (-1)^{\sum_{i=1}^n b_i x_i}
\end{equation}
is the Fourier transform of $p$ over $\ZZ_2^n$.
\end{theorem}
\endgroup

It follows directly from the above theorem (Corr.~\ref{corr:general-eigenvalue-bound})
that the \NAME ensemble is tomographically complete if the distribution $p$ is strictly positive, i.e.,\ $p(s) > 0$ for all $s \in \{0, 1\}^n$. In this case, the shadow channel is invertible and the inverse is defined by
\begin{equation}
    \MMM^{-1}(P_{a,b}) = \lambda_{a,b}^{-1} \, P_{a,b}.
\end{equation}
The exact eigenvalues $\lambda_{a,b}$ of the shadow channel can either be computed explicitly for concrete examples (Sec.~\ref{sec:example-distributions}) or numerically for arbitrary distributions $p$ using a fast Hadamard transform.

The number of samples required to estimate the expectation value of an observable $A$ using classical shadows depends on the shadow norm $\norm{A}_{\mathrm{sh}}$ with respect to the underlying ensemble. In the case of the local Pauli ensemble and the global Clifford ensemble, their shadow norms can be bounded by
\begin{equation}
    \norm{A}_{\mathrm{sh,P}}^2
    \leq
    4^{k} \cdot \norm{A}^2_\infty,
    \qquad
    \norm{A}_{\mathrm{sh,C}}^2
    \leq
    3 \cdot \Tr(A^2)
    \leq
    3 \cdot 2^n \cdot \norm{A}^2_\infty,
\end{equation}
respectively, where $k$ denotes the locality of $A$. In particular, the local Pauli ensemble is efficient for $k$-local observables, while the global Clifford ensemble is more efficient for general or low-rank observables.

We call an observable $S$-sparse for a subset $S \subseteq \{0,1\}^n$ if it is supported only on Pauli operators $P_{a,b}$ with $a \in S$. Our second main result provides a general bound on the shadow norm of the \NAME ensemble for strictly positive distributions $p$ and $S$-sparse observables.

\begingroup
\let\thetheorem2
\begin{theorem}[Thm.~\ref{thm:main-shadow-bound}]
Let $S \subseteq \{0,1\}^n$ be non-empty and $A$ be an $S$-sparse Hermitian operator. Then the shadow norm of the \NAME ensemble with a strictly positive distribution $p$ is bounded by
\begin{equation}
    \norm{A}_{\mathrm{sh}}^2
    \leq
    (6 + \pi) \mleft[\max_{s \in S} \frac{1}{p(s)}\mright] \norm{A}^2_\infty.
\end{equation}
\end{theorem}
\endgroup

The proof of the above theorem relies on discrete Fourier analysis and provides a novel application of Markov semigroup theory to bounding shadow norms. In particular, we show that the shadow norm of diagonal operators can be expressed in terms of a Dirichlet form on a double-cover of $\ZZ_2^n$, which can then be bounded using a reverse Poincaré inequality formulated in terms of the carré du champ operator of the heat semigroup on this graph (Sec.~\ref{sec:proof}).

In the case of the uniform distribution $p(x) = \frac{1}{2^n}$, the above theorem implies that the shadow norm of an arbitrary Hermitian operator $A$ is bounded by
\begin{equation}
    \norm{A}_{\mathrm{sh}}^2
    \leq
    (6 + \pi) \cdot 2^n \cdot \norm{A}^2_\infty.
\end{equation}
This bound agrees with the worst-case bound for the global Clifford ensemble up to a constant factor, and this scaling is optimal in the worst case for shadow protocols based on single-copy measurements; see~\cite[Thm.~2]{huang2020predicting}. However, for $S$-sparse observables with $|S| \ll 2^n$, significant improvements are possible (Ex.~\ref{ex:uniform-on-S}).

A special class of $S$-sparse operators consists of all $k$-$X/Y$-local operators, where each Pauli operator in the support has at most $k$ tensor factors belonging to $\{X, Y\}$. This class includes $k$-local operators, but also other physically relevant observables arising, for example, in the context of fermionic modes under the Jordan--Wigner transformation~\cite{zhao2021fermionic} or from long-range Kitaev chains~\cite{vodola2014kitaev}.

By choosing a subset from a binomial distribution with parameter $q = \frac{1}{n+1}$, we obtain the following shadow-norm bound for $k$-$X/Y$-local operators (Ex.~\ref{ex:binomial-shadow-norm}):
\begin{equation}
    \norm{A}_{\mathrm{sh}}^2
    \leq
    (6 + \pi) \cdot e \cdot n^k \cdot \norm{A}^2_\infty.
\end{equation}
This bound is polynomial in $n$ for fixed $k$, which is a significant improvement over the exponential worst-case scaling of the global Clifford ensemble bound. While our scaling is worse than $4^k$ for $k$-local operators using the local Pauli ensemble, we can obtain an exponential improvement over the local Pauli ensemble bound for observables such as $Z \otimes \dots \otimes Z$ that are $k$-$X/Y$-local but not $k$-local. If the $X/Y$-locality of the observable is known in advance, the parameter $q$ can be adjusted to further optimize the constant factor (Ex.~\ref{ex:binomial-tuned}).
Moreover, the same polynomial guarantee applies to the complementary $X/Y$-locality class with at most $k$ tensor factors lying in $\{I,Z\}$, which contains $\boldsymbol{\Theta}(2^n)$ non-trivial Pauli observables.

We numerically investigate the corresponding sample complexity for several values of $k$, using an appropriately chosen $q$. We further demonstrate favorable sample complexity for estimating $k$-body reduced density matrices of $n$-mode fermionic states. This setting also illustrates the flexibility of \NAME shadows, as adapting the subset distribution $p$ to the structure of the target observables substantially improves practical performance.

\subsection{Overview}\label{sec:overview}

The remainder of the paper is organized as follows.
Section~\ref{sec:preliminaries} introduces notation and the
classical shadow formalism, before recalling basics from the theory of Clifford shadows and Markov semigroups.
Section~\ref{sec:ensemble} defines the \NAME ensemble and finds its shadow channel.
Section~\ref{sec:channel-spectrum} computes the exact spectrum of the
\NAME shadow channel (Theorem~\ref{thm:main-channel-eigenvalues}) and illustrates it
for several choices of probability distributions. Section~\ref{sec:shadow-norm} develops
shadow-norm bounds for structured observables, with the proof of the
main bound (Theorem~\ref{thm:main-shadow-bound}) given in
Section~\ref{sec:proof}, and studies the special class of $X/Y$-local observables. Section~\ref{sec:experiments} reports
numerical experiments validating the predicted scaling, and
Section~\ref{sec:discussion} discusses these results in the context of
related work and outlines open questions.

\section{Preliminaries}\label{sec:preliminaries}

We begin by introducing some notation and basic definitions.
Throughout the paper, let $n \geq 1$ be an integer and consider $n$ qubits on the Hilbert space $H = (\CC^2)^{\otimes n}$ with
computational basis $\ket{x}$ for $x \in \{0, 1\}^n$. If $A$ is a linear operator on $H$, we denote by $\norm{A}_\infty$ its spectral norm and by $\norm{A}_1$ its trace norm.
Moreover, for two Hermitian operators $A$ and $B$, we write $A \preceq B$ if $B - A$ is positive semidefinite.

Consider the finite abelian group $\ZZ_2^n$. We identify elements of $\ZZ_2^n$ with bitstrings in $\{0, 1\}^n$ and
write $+$ for bitwise addition modulo $2$. Additionally, we write $\zero = (0, \dots, 0)$, $\one = (1, \dots, 1)$, and denote by $\wt(x) = \sum_{i=1}^n x_i$ the \emph{Hamming weight} of $x \in \{0, 1\}^n$.
Let $a, b \in \{0, 1\}^n$. Then we define their ordinary integer dot product $a \cdot b = \sum_{i=1}^n a_i b_i$ and the \emph{characters}
\begin{equation}\label{qe:characters}
    \chi_a \colon \{0,1\}^n \to \{\pm 1\}, \quad
    \chi_a(x) = (-1)^{a \cdot x}.
\end{equation}
Note that characters satisfy various algebraic properties, such as
\begin{equation}\label{qe:character-properties}
    \chi_a(b) = \chi_b(a), \qquad
    \chi_{\zero}(a) = 1, \qquad
    \chi_{a+b}(c) = \chi_a(c) \chi_b(c),
\end{equation}
and are orthogonal with respect to averaging, i.e.,
\begin{equation}\label{qe:character-orthogonality}
    \frac{1}{2^n}\sum_{x \in \{0,1\}^n} \chi_a(x) \chi_b(x) = \begin{cases}
        1, & \text{if $a = b$}, \\
        0, & \text{otherwise}.
    \end{cases}
\end{equation}
For more information and the group-theoretic background, we refer to~\cite{serre1977linear}.

Throughout the paper, we denote by $I, X, Y, Z$ the single-qubit Pauli matrices
\begin{equation}
    I = \begin{pmatrix}
        1 & 0 \\
        0 & 1
    \end{pmatrix},
    \quad
    X = \begin{pmatrix}
        0 & 1 \\
        1 & 0
    \end{pmatrix},
    \quad
    Y = \begin{pmatrix}
        0 & -\ii \\
        \ii & 0
    \end{pmatrix},
    \quad
    Z = \begin{pmatrix}
        1 & 0 \\
        0 & -1
    \end{pmatrix},
\end{equation}
and we define the $n$-qubit \emph{Pauli group}
\begin{equation}
    \PPP_n = \{ \pm 1, \pm \ii \} \cdot \{I, X, Y, Z\}^{\otimes n}.
\end{equation}
We say that a Pauli operator $P \in \PPP_n$ is of \emph{$Z$-type} if $P \in \{\pm 1, \pm \ii\} \cdot \{I, Z\}^{\otimes n}$, or equivalently, if $P$ is diagonal in the computational basis.
For any $a, b \in \{0, 1\}^n$, we define the Pauli operator
\begin{equation}\label{eq:symmplectic-Pauli}
    P_{a,b} = \bigotimes_{k=1}^n \ii^{a_k b_k} X^{a_k} Z^{b_k}.
\end{equation}
Then each $P_{a,b}$ is Hermitian and the set $\{P_{a,b}\}_{a,b \in \{0, 1\}^n}$ agrees with the
standard Pauli basis $\{I, X, Y, Z\}^{\otimes n}$. In particular, any operator $A$ can be expressed as
\begin{equation}
    A = \sum_{a, b \in \{0, 1\}^n} \alpha_{a,b} P_{a,b}, \qquad
    \alpha_{a,b} = \frac{1}{2^n} \Tr(P_{a,b} A).
\end{equation}
Moreover, the action of $P_{a,b}$ on computational basis states is given by
\begin{equation}\label{eq:Pauli-basis-action}
    P_{a,b} \ket{x} = \ii^{a \cdot b} \, \chi_{b}(x) \, \ket{x + a},
\end{equation}
and the product of two Pauli operators can be expressed as
\begin{equation}\label{eq:Pauli-product}
    P_{a,b} \cdot P_{c,d}
    = \ii^{a \cdot b} \, \ii^{c \cdot d} \, (-\ii)^{(a + c) \cdot (b + d)} \, \chi_{b}(c) \, P_{(a + c, b + d)}.
\end{equation}
Thus, the mapping $(a, b) \to P_{a,b}$ respects the group structure up to phase factors.
See~\cite{calderbank1997quantum,aaronson2004improved} for more details on this symplectic representation of Pauli operators
and its applications in quantum error correction and in the context of stabilizer circuits. For a general introduction to quantum information, we refer to~\cite{nielsen2010quantum,watrous2018theory}.

\subsection{Classical shadows}\label{sec:classical-shadows}

In the following, we introduce the basic concepts of classical shadows developed in~\cite{huang2020predicting}. We will keep the presentation short and focus on the technical aspects that are relevant to our work. For more details on applications of classical shadows, the use of different unitary ensembles, and other related work, we refer to Section~\ref{sec:related-work}.

The basic idea of classical shadows is to compute a classical representation of a quantum state $\rho$ by performing randomized measurements on multiple copies of $\rho$.
These measurements are defined by a \emph{unitary ensemble} $\{U_i\}_{i \in I}$ together with a probability distribution $p \colon I \to [0, 1]$. In the following, we
restrict ourselves to ensembles where the index set $I$ is countable and the probability distribution is discrete.

Two examples of such unitary ensembles, both studied in~\cite{huang2020predicting}, are the
\emph{local Pauli ensemble}, in which each qubit is measured independently and uniformly in one of the three Pauli eigenbases, and the
\emph{global Clifford ensemble}, in which a Clifford
unitary modulo global phase is sampled uniformly before measuring in the computational basis. Given a unitary ensemble $\{U_i\}_{i \in I}$ with a corresponding probability distribution, its \emph{shadow channel} $\MMM$ is defined by
\begin{equation}\label{eq:shadow-channel}
    \MMM(\rho) = \EE_{i \in I} \mleft[ \sum_{x \in \{0, 1\}^{n}} \bra{x} U_i \rho U_i^\dagger \ket{x} \cdot U_i^\dagger \ketbra{x} U_i\mright].
\end{equation}
We call a unitary ensemble \emph{tomographically complete} if the corresponding shadow channel $\MMM$ is invertible. In this case, we
can run the classical shadow protocol (Algorithm~\ref{alg:classical-shadow}) to compute a classical shadow $\widehat{\rho}$ of a quantum state $\rho$.

\begin{algorithm}[htbp]
  \caption{Classical shadow protocol}
  \label{alg:classical-shadow}
  \DontPrintSemicolon
  \SetAlgoLined

  \KwIn{%
  \begin{tabular}[t]{@{}l@{}}
    $N$ copies of a quantum state $\rho$, \\
    a unitary ensemble $\{U_i\}_{i \in I}$ together with a probability distribution.
  \end{tabular}%
    }
\BlankLine
  \KwOut{Classical shadow $\widehat{\rho}$ of size $N$.}

  \BlankLine
  \For{$j \gets 1$ \KwTo $N$}{
    Sample $i \in I$ according to the probability distribution.\;

    Apply $U_i$ to a fresh copy of $\rho$.\;

    Measure in the computational basis to obtain $x_j$.\;

    Compute $
      \widehat{\rho}_j
      =
      \MMM^{-1}\!\mleft(
        U_i^\dagger
        \ketbra{x_j}{x_j}
        U_i
      \mright)
    $.\;
  }
  \BlankLine
  \Return{$\widehat{\rho}
    = (\widehat{\rho}_1,\ldots,\widehat{\rho}_N)$}\;
\end{algorithm}

Note that we need to be able to efficiently sample from the unitary ensemble, compute the inverse of the shadow channel $\MMM^{-1}$ and store the resulting classical shadow $\widehat{\rho}$ in classical memory in order to perform the classical shadow protocol efficiently.
In this case, we can use a median-of-means estimator in a second step to efficiently estimate expectation values of observables $A$ with respect to the quantum state $\rho$ without having to perform full quantum state tomography.

\begin{proposition}[\cite{huang2020predicting}]\label{prop:shadow-norm-sample-complexity}
Let $\rho$ be a quantum state and consider $M$ observables $A_1, \dots, A_M$. Then the classical shadow $\widehat{\rho}$ of size $N$ can be used to predict the expectations
$\Tr(A_1\rho)$, \dots, $\Tr(A_M\rho)$ up to additive error $\varepsilon > 0$ with probability at least $1 - \delta$ given that
\begin{equation}\label{eq:shadow-norm-sample-complexity}
    N \gtrsim \log(\frac{M}{\delta}) \cdot \frac{1}{\varepsilon^2} \cdot \max_{i=1,\ldots,M} \|A_i\|_{\mathrm{sh}}^2.
\end{equation}
\end{proposition}

The quantity $\|A\|_{\mathrm{sh}}$ in Eq.~\eqref{eq:shadow-norm-sample-complexity} is called the \emph{shadow norm} of $A$ and depends on the unitary ensemble $\{U_i\}_{i \in I}$
as well as the underlying probability distribution. Denote by $\DDD$ the set of all density matrices. Then the shadow norm is defined by
\begin{equation}\label{eq:shadow-norm}
    \norm{A}_{\mathrm{sh}}^2 = \sup_{\sigma \in \DDD} \EE_{i\in I} \mleft[\sum_{x \in \{0, 1\}^{n}} \bra{x}U_i \sigma U_i^\dagger \ket{x} \bigl(\bra{x}U_i \MMM^{-1}(A) U_i^\dagger \ket{x}\bigr)^2\mright].
\end{equation}

By introducing auxiliary channels $\MMM_i$, both the shadow channel and the shadow norm can be expressed in a more compact way. Moreover, these auxiliary channels will be useful in the analysis of our
\NAME shadow channel in Section~\ref{sec:channel-spectrum} and Section~\ref{sec:shadow-norm}. Let $\{U_i\}_{i \in I}$ be a unitary ensemble.
Then we define the quantum channels $\MMM_i$ by
\begin{equation}\label{eq:aux-channels}
    \MMM_i(\rho) = \sum_{x \in \{0, 1\}^{n}} \bra{x} U_i \rho U_i^\dagger \ket{x} U_i^\dagger \ketbra{x} U_i.
\end{equation}
It follows directly from Eq.~\eqref{eq:shadow-channel} that the shadow channel can be expressed as an expectation
\begin{equation}
    \MMM(\rho) = \EE_{i \in I} \bigl[ \MMM_i(\rho) \bigr].
\end{equation}
Additionally, the following proposition shows that the shadow norm can also be expressed more compactly in terms of the channels $\MMM_i$.

\begin{proposition}\label{prop:shadow-norm-short}
Let $\{U_i\}_{i \in I}$ be a tomographically complete unitary ensemble and $A$ be a Hermitian operator. Then the shadow norm can be written as
\begin{equation}
    \norm{A}_{\mathrm{sh}}^2 = \sup_{\sigma \in \DDD} \Tr(\sigma \EE_{i \in I} \bigl[ \MMM_i(\MMM^{-1}(A))^2 \bigr] ).
\end{equation}
\end{proposition}
\begin{proof}
According to Eq.~\eqref{eq:shadow-norm}, we have
\begin{equation}
    \norm{A}_{\mathrm{sh}}^2 = \sup_{\sigma \in \DDD} \EE_{i \in I} \mleft[\sum_{x \in \{0, 1\}^{n}} \bra{x}U_i \sigma U_i^\dagger \ket{x} \bigl(\bra{x}U_i \MMM^{-1}(A) U_i^\dagger \ket{x}\bigr)^2\mright].
\end{equation}
Using the fact that $\bra{x}U_i \sigma U_i^\dagger \ket{x} = \Tr(\sigma U_i^\dagger \ketbra{x} U_i)$, we can rewrite the shadow norm as
\begin{align*}
    \norm{A}_{\mathrm{sh}}^2 &= \sup_{\sigma \in \DDD} \Tr(\sigma \EE_{i \in I} \mleft[ \sum_{x \in \{0, 1\}^{n}}  U_i^\dagger \ketbra{x} U_i \bigl(\bra{x}U_i \MMM^{-1}(A) U_i^\dagger \ket{x}\bigr)^2 \mright]).
\end{align*}
Since the projections $U_i^\dagger \ketbra{x} U_i$ are orthogonal, it follows that
\begin{multline}
    \sum_{x \in \{0, 1\}^{n}} U_i^\dagger \ketbra{x} U_i \bigl(\bra{x}U_i \MMM^{-1}(A) U_i^\dagger \ket{x}\bigr)^2 \\
    = \mleft( \sum_{x \in \{0, 1\}^{n}} U_i^\dagger \ketbra{x} U_i \bra{x}U_i \MMM^{-1}(A) U_i^\dagger \ket{x} \mright)^2
    = \MMM_{i}(\MMM^{-1}(A))^2.
\end{multline}
Therefore, we have
\begin{equation}
    \norm{A}_{\mathrm{sh}}^2 =\sup_{\sigma \in \DDD} \Tr(\sigma \EE_{i \in I} \bigl[ \MMM_{i}(\MMM^{-1}(A))^2 \bigr] ).
\end{equation}
\end{proof}

\subsection{Clifford ensembles}\label{sec:clifford-ensembles}

Recall that the $n$-qubit \emph{Clifford group} $\CCC_n$ is given by the normalizer of the Pauli group inside the unitary group, i.e.,
\begin{equation}\label{eq:clifford-group}
    \CCC_n = \{ U \colon U \PPP_n U^\dagger = \PPP_n \}.
\end{equation}
We call a unitary ensemble $\{U_i\}_{i \in I}$ a \emph{Clifford ensemble} if $U_i \in \CCC_n$ for all $i \in I$.

In the following, we state some general facts about classical shadows defined by Clifford unitaries. In particular, we show that, in this setting,
the shadow channel is diagonal in the Pauli basis and the shadow norm of a Pauli operator can be expressed in terms of its corresponding eigenvalue.
These properties are well-known and appear throughout the classical shadow literature in the context of Pauli and Clifford shadows; see, e.g.,~\cite{huang2020predicting,bertoni2024shallow,bu2024classical}. However, we give short proofs of all the statements below to keep the presentation self-contained.

\begin{proposition}\label{prop:M-eigenvalues}
Let $\{U_i\}_{i \in I}$ be a Clifford unitary ensemble and $P \in \PPP_n$ be a Pauli operator. Then
\begin{equation}
    \MMM_{i}(P) = \begin{cases}
        P, & \text{if $U_i P U_i^\dagger$ is of $Z$-type}, \\
        0, & \text{otherwise}.
    \end{cases}
\end{equation}
Moreover, the shadow channel $\MMM$ is diagonal in the Pauli basis, i.e.,
\begin{equation}
    \MMM(P) = \lambda_P P, \qquad \lambda_P = \PP_{i \in I}[\text{$U_i P U_i^\dagger$ is of $Z$-type}].
\end{equation}
\end{proposition}
\begin{proof}
Since $U_i$ is a Clifford unitary, there exists a Pauli operator
$Q \in \PPP_n$ such that $U_i P U_i^\dagger = Q$. Hence, we can write
\begin{equation}
    \MMM_{i}(P)
    = \sum_{x \in \{0,1\}^n} \bra{x} U_i P U_i^\dagger \ket{x}
    U_i^\dagger \ketbra{x} U_i
    = \sum_{x \in \{0,1\}^n} \bra{x} Q \ket{x}
    U_i^\dagger \ketbra{x} U_i.
\end{equation}
If $Q$ is of $Z$-type, then $Q$ is diagonal in the computational basis. Thus, we have
\begin{equation}
    Q = \sum_{x \in \{0, 1\}^{n}} \bra{x} Q \ket{x} \ketbra{x},
\end{equation}
which implies
\begin{align}
    \MMM_{i}(P) = U_i^\dagger \Bigl( \sum_{x \in \{0, 1\}^{n}} \bra{x} Q \ket{x} \ketbra{x} \Bigr) U_i
    = U_i^\dagger Q U_i = P.
\end{align}
Otherwise, if $Q$ is not of $Z$-type, then $\bra{x} Q \ket{x} = 0$ for all $x \in \{0, 1\}^{n}$, and we have $\MMM_{i}(P) = 0$.
This proves the first part of the statement. For the second part, we compute
\begin{equation}
    \MMM(P) = \EE_{i \in I}\bigl[\MMM_{i}(P)\bigr] = \PP_{i \in I}[\text{$U_i P U_i^\dagger$ is of $Z$-type}] \cdot P.
\end{equation}
\end{proof}

\begin{proposition}\label{prop:shadow-norm-single-pauli}
Let $\{U_i\}_{i \in I}$ be a tomographically complete Clifford ensemble and $P \in \{I,X,Y,Z\}^{\otimes n}$. Then the shadow norm of $P$ is given by
\begin{equation}
    \norm{P}_{\mathrm{sh}}^2 = \lambda_P^{-1}.
\end{equation}
\end{proposition}
\begin{proof}
By Proposition~\ref{prop:M-eigenvalues}, the shadow channel $\MMM$ is diagonal in the Pauli basis
with $\lambda_P > 0$ since the ensemble is tomographically complete. Thus, $\MMM^{-1}(P) = \lambda_P^{-1} P$ and we can write
\begin{equation}
    \EE_{i \in I} \bigl[ \MMM_{i}(\MMM^{-1}(P))^2 \bigr] = \lambda_P^{-2} \cdot \EE_{i \in I} \bigl[ \MMM_{i}(P)^2\bigr].
\end{equation}
Furthermore, Proposition~\ref{prop:M-eigenvalues} implies $\MMM_{i}(P)^2 = \MMM_{i}(P) \cdot P$. Hence,
\begin{equation}
    \lambda_P^{-2} \cdot \EE_{i \in I} \bigl[ \MMM_{i}(P)^2 \bigr] = \lambda_P^{-2} \cdot \EE_{i \in I} \bigl[ \MMM_{i}(P)\bigr] \cdot P = \lambda_P^{-2} \cdot \MMM(P) \cdot P = \lambda_P^{-1} P^2.
\end{equation}
Substituting the previous equation into Proposition~\ref{prop:shadow-norm-short}, and using $P^2 = I^{\otimes n}$, yields
\begin{equation}
    \norm{P}_{\mathrm{sh}}^2 = \lambda_P^{-1} \sup_{\sigma \in \DDD} \Tr(\sigma) = \lambda_P^{-1}.
\end{equation}
\end{proof}

Alternatively, the $Z$-type condition in Proposition~\ref{prop:M-eigenvalues} can be reformulated in terms of eigenvectors of the corresponding Pauli operator.
We will use this reformulation in the proof of Theorem~\ref{thm:main-channel-eigenvalues} to compute the eigenvalues of our \NAME shadow channel.

\begin{proposition}\label{prop:UPU-alt}
Let $U \in \CCC_n$ be a Clifford unitary and $P \in \PPP_n$ be a Pauli operator. The following are equivalent:
\begin{enumerate}
    \item $U P U^\dagger$ is of $Z$-type,
    \item $U^\dagger \ket{x}$ is an eigenvector of $P$ for all $x \in \{0, 1\}^{n}$.
\end{enumerate}
\end{proposition}
\begin{proof}
Since $U$ is Clifford, there exists a Pauli operator $Q \in \PPP_n$ such that $U P U^\dagger = Q$. Assume that $U P U^\dagger$ is of $Z$-type. Then $Q \ket{x} = \lambda_x \ket{x}$ for all $x \in \{0, 1\}^{n}$, which yields
\begin{equation}
    P U^\dagger \ket{x} = U^\dagger Q \ket{x} = \lambda_x U^\dagger \ket{x}.
\end{equation}
Conversely, assume that $P U^\dagger \ket{x} = \lambda_x U^\dagger \ket{x}$ for all $x \in \{0, 1\}^{n}$. Then we have
\begin{equation}
    U P U^\dagger \ket{x} = \lambda_x U U^\dagger \ket{x} = \lambda_x \ket{x}.
\end{equation}
Thus, $U P U^\dagger$ is diagonal in the computational basis, or
equivalently, is of $Z$-type.
\end{proof}

\subsection{Markov semigroups}\label{sec:markov-semigroups}

In this section, we introduce some basic definitions and facts from the theory of Markov semigroups, including an integral representation of the resolvent and a reverse Poincaré inequality.
These results will be used in Lemma~\ref{lem:graph-energy-estimate} as part of the proof of Theorem~\ref{thm:main-shadow-bound}. In the following, we restrict to the finite-dimensional case and provide short proofs of all propositions
for completeness. For a more detailed introduction to Markov semigroups, or more generally one-parameter continuous semigroups, we refer to~\cite{bakry2014analysis,engel2000one-parameter}.

Throughout the section, let $X$ be a finite set and denote by $\ell^\infty(X)$ the space of all real-valued functions $f \colon X \to \RR$ equipped with the supremum norm $\norm{\cdot}_\infty$.

\begin{definition}
A linear operator $P \colon \ell^\infty(X) \to \ell^\infty(X)$ is a \emph{Markov operator} if
\begin{enumerate}
    \item $P 1 = 1$, where $1 \in \ell^\infty(X)$ denotes the constant function $x \mapsto 1$,
    \item $f \geq 0$ implies $P f \geq 0$.
\end{enumerate}
\end{definition}
It follows directly from the second condition that $Pf \leq Pg$ whenever $f \leq g$. In addition, Markov operators are contractions
with respect to the supremum norm and satisfy the following Cauchy--Schwarz-type inequality.

\begin{proposition}
Let $P$ be a Markov operator and $f \colon X \to \RR$. Then
\begin{equation}
    \norm{Pf}_\infty \leq \norm{f}_\infty,
    \qquad
    (P f)^2 \leq P(f^2).
\end{equation}
\end{proposition}
\begin{proof}
Note that for every $x \in X$, the function $B_x(f, g) = P(f \cdot g)(x)$ defines a
positive-semidefinite symmetric bilinear form. Hence, the Cauchy--Schwarz inequality yields
\begin{equation}
    (P f)^2(x) = B_x(f, 1)^2 \leq B_x(f, f) \cdot B_x(1, 1) = P(f^2)(x) \cdot P(1^2)(x).
\end{equation}
Using $P 1 = 1$ proves the second inequality. Furthermore,
$f \leq \norm{f}_\infty \cdot 1$ and the Markov property imply
\begin{equation}
    (P f)^2 \leq P(f^2) \leq P(\norm{f}_\infty^2 \cdot 1) = \norm{f}_\infty^2 \cdot 1.
\end{equation}
Taking square roots and the supremum over all $x \in X$ proves the first inequality.
\end{proof}

Next we define Markov semigroups and their infinitesimal generators.

\begin{definition}
A family $\{P_t\}_{t \geq 0}$ of Markov operators is a \emph{Markov semigroup} if
\begin{enumerate}
\item $P_0 = \id$ and $P_{s+t} = P_s P_t$,
\item $\lim_{t \downarrow 0} \norm{P_t f - f}_\infty = 0$ for all $f \in \ell^\infty(X)$.
\end{enumerate}
\end{definition}
The first condition states that the family is a semigroup, while the second condition implies the continuity of the map
$t \mapsto P_t$ in operator norm. Let $L$ be a linear operator. Then the family $\{e^{tL}\}_{t \geq 0}$ is always a (not necessarily Markov)
continuous semigroup. Moreover, the following proposition shows that the converse also holds, i.e.,\ every Markov semigroup is of this form.

\begin{proposition}\label{prop:Markov-exp-L}
Let $\{P_t\}_{t \geq 0}$ be a Markov semigroup. Then there exists a linear operator $L$ such that $P_t = e^{tL}$ for all $t \geq 0$.
\end{proposition}
\begin{proof}
See~\cite[Theorem I.2.9]{engel2000one-parameter} for a proof in the setting of general matrix semigroups.
\end{proof}

The operator $L$ in the previous proposition is called the \emph{infinitesimal generator} of the Markov semigroup.
Using matrix-valued calculus, we obtain the following integral representation of its resolvent.

\begin{proposition}[Resolvent representation]\label{prop:resolvent-integral}
Let $\{P_t\}_{t \geq 0}$ be a Markov semigroup with infinitesimal generator $L$. Then it holds for every $\lambda > 0$ that
\begin{equation}
    (\lambda \id - L)^{-1}
    =
    \int_0^\infty
    e^{-\lambda t}P_t \, \dd t.
\end{equation}
\end{proposition}
\begin{proof}
Let $\lambda > 0$ and define the operator
\[
    R = \int_0^\infty e^{-\lambda t}P_t \, \dd t.
\]
Since each $P_t$ is a contraction, the integral
converges and $e^{-\lambda t}P_t \to 0$ as $t \to \infty$.
From Proposition~\ref{prop:Markov-exp-L}, it follows that $\frac{\dd}{\dd t}P_t = L P_t$.
Hence, we have
\[
    (\lambda\id-L)R
    = -\int_0^\infty
      \frac{\dd}{\dd t}\bigl(e^{-\lambda t}P_t\bigr)\, \dd t
    = -(0 - P_0)
    = \id.
\]
Since all operators are finite-dimensional, this is equivalent to $R = (\lambda\id-L)^{-1}$.
\end{proof}

Consider a Markov semigroup with infinitesimal generator $L$. Then its \emph{carré du champ} $\Gamma \colon \ell^\infty(X) \to \ell^\infty(X)$ is defined by
\begin{equation}
    \Gamma(f) = \frac{1}{2} \bigl( L(f^2) - 2 f L(f) \bigr).
\end{equation}
The next proposition establishes a reverse Poincaré inequality for certain Markov semigroups following the argument in~\cite[Proposition 3]{ledoux2011concentration}.

\begin{proposition}[Reverse Poincaré inequality]\label{prop:rev-poincare-inequ}
Let $\{P_t\}_{t \geq 0}$ be a Markov semigroup satisfying $\Gamma(P_t f) \leq P_t(\Gamma f)$ for all $f \in \ell^\infty(X)$ and $t \geq 0$. Then
\begin{equation}
    \Gamma(P_t f)
    \leq
    \frac{1}{2t} \, \norm{f}_\infty^2
    \qquad (t>0).
\end{equation}
\end{proposition}
\begin{proof}

Let $f \colon X \to \RR$ and $t > 0$. For $s \in [0, t]$, define $g = P_{t-s}f$. Then we have
\begin{equation}
    \frac{\dd}{\dd s} g = -Lg,
    \qquad
    \frac{\dd}{\dd s} P_s g^2
    = LP_s g^2 - 2 P_s (g Lg)
    = 2 P_s(\Gamma g),
\end{equation}
where we used for the last equality that $P_s$ is a matrix exponential in $L$ and thus commutes with $L$. Hence, we can write
\begin{equation}
    P_t(f^2) - (P_t f)^2 = \int_0^t \frac{\dd}{\dd s} P_s(P_{t-s}f)^2 \, \dd s
    = 2 \int_0^t P_s(\Gamma(g)) \, \dd s.
\end{equation}
Using the assumption $\Gamma(P_t f) \leq P_t(\Gamma f)$ and $P_s g = P_t f$, we obtain
\begin{equation}
    P_t(f^2) - (P_t f)^2
    \geq 2 \int_0^t \Gamma(P_t f) \, \dd s
    = 2 t \, \Gamma(P_t f).
\end{equation}
Dividing by $2t$ and using the Markov property of $P_t$ yields
\begin{equation}
    \Gamma(P_t f)
    \leq
    \frac{1}{2t} \, P_t(f^2)
    \leq
     \frac{1}{2t} \, P_t\bigl(\norm{f}_\infty^2 \cdot 1\bigr)
    =
    \frac{1}{2t} \, \norm{f}_\infty^2 \cdot 1.
\end{equation}
\end{proof}

\section{The selective GHZ ensemble}\label{sec:ensemble}

In~\cite{patel2026selective}, the authors introduce a new measurement scheme for quantum state tomography called \emph{Selective and Efficient Quantum State Tomography (SEEQST)}.
The main idea of this scheme is to partition a density matrix into $2^n$ disjoint blocks so that each block can be estimated using only two measurement settings. In the following,
we introduce a unitary ensemble based on these measurement settings that will give rise to a corresponding classical shadow protocol.

To define the unitary ensemble, we first introduce the \emph{Greenberger--Horne--Zeilinger (GHZ) states} indexed by $k \in \{0, 1\}$ and $x \in \{0, 1\}^m$:
\begin{equation}\label{eq:GHZ-states}
    \ket{\varphi_{k,x}} = \frac{1}{\sqrt{2}} \Bigl( \ket{0, x_2, \dots, x_m} + (-1)^{x_1} \ii^k \ket{1, \overline{x_2}, \dots, \overline{x_m}}
    \Bigr).
\end{equation}
These states generalize the standard GHZ states in~\cite{greenberger1989going}. Additionally, we introduce the corresponding \emph{even and odd GHZ bases}
\begin{equation}\label{eq:GHZ-bases}
    \GGG_{0,m} = \bigl\{\ket{\varphi_{0,x}}\bigr\}_{x \in \{0, 1\}^m},
    \qquad
    \GGG_{1,m} = \bigl\{\ket{\varphi_{1,x}}\bigr\}_{x \in \{0, 1\}^m}.
\end{equation}
Using this notation, each setting performs a measurement in either the even or the odd GHZ basis on a subset of qubits while measuring the remaining qubits in the computational basis. By translating these measurements into
unitary transformations, we arrive at the following definition of our unitary ensemble.

\begin{definition}[\NAME ensemble]\label{def:main-ensemble}
Let $p \colon \{0,1\}^n \to [0, 1]$ be a probability distribution. Then the $n$-qubit \emph{selective GHZ measurement (\NAME) ensemble} with distribution $p$ is given by
\begin{equation}
    \bigl\{ U_{k,s} : k \in \{0, 1\}, s \in \{0, 1\}^n \bigr\},
\end{equation}
where the unitary $U_{k, s}$ is chosen with probability $\frac{1}{2} \, p(s)$ and transforms all qubits $i$ with $s_i = 1$ from the corresponding $\GGG_{k,m}$ basis into the computational basis.
\end{definition}

Note that computational basis states can easily be transformed into even or odd GHZ states using the following combination of Hadamard, phase, and controlled-NOT gates.
\begin{equation}
\scalebox{0.8}{
\begin{quantikz}
\lstick{$\ket{x_1}$}
    & \gate{H}
    & \ctrl{1}
    & \qw
    & \cdots
    & \ctrl{3}
    & \qw \rstick[5]{$\ket{\varphi_{0,x}}$} \\
\lstick{$\ket{x_2}$}
    & \qw
    & \targ{}
    & \qw
    & \cdots
    & \qw
    & \qw \\
\lstick{$\cdots$}
    & \qw
    & \qw
    & \qw
    & \cdots
    & \qw
    & \qw \\
\lstick{$\ket{x_n}$}
    & \qw
    & \qw
    & \qw
    & \cdots
    & \targ{}
    & \qw
\end{quantikz}}
\qquad
\scalebox{0.8}{
\begin{quantikz}
\lstick{$\ket{x_1}$}
    & \gate{H}
    & \gate{S}
    & \ctrl{1}
    & \cdots
    & \ctrl{3}
    & \qw \rstick[5]{$\ket{\varphi_{1,x}}$} \\
\lstick{$\ket{x_2}$}
    & \qw
    & \qw
    & \targ{}
    & \cdots
    & \qw
    & \qw \\
\lstick{$\cdots$}
    & \qw
    & \qw
    & \qw
    & \cdots
    & \qw
    & \qw \\
\lstick{$\ket{x_n}$}
    & \qw
    & \qw
    & \qw
    & \cdots
    & \targ{}
    & \qw
\end{quantikz}
}
\end{equation}
Hence, every unitary $U_{k, s}$ can be implemented by applying one of the circuits above in reverse to all qubits $i$ with $s_i = 1$.
For example, the following circuits implement the unitaries $U_{0, 1011}$ and $U_{1, 1011}$ that transform the qubits $1$, $3$ and $4$ from the even and odd GHZ basis into the computational basis, respectively.
\begin{equation}
U_{0, 1011} \colon\scalebox{0.8}{
\begin{quantikz}
    & \ctrl{3}
    & \ctrl{2}
    & \gate{H}
    & \qw \\
    & \qw
    & \qw
    & \qw
    & \qw \\
    & \qw
    & \targ{}
    & \qw
    & \qw \\
    & \targ{}
    & \qw
    & \qw
    & \qw
\end{quantikz}
}
\qquad\quad
U_{1, 1011} \colon\scalebox{0.8}{
\begin{quantikz}
    & \ctrl{3}
    & \ctrl{2}
    & \gate{S^\dagger}
    & \gate{H}
    & \qw \\
    & \qw
    & \qw
    & \qw
    & \qw
    & \qw \\
    & \qw
    & \targ{}
    & \qw
    & \qw
    & \qw \\
    & \targ{}
    & \qw
    & \qw
    & \qw
    & \qw
\end{quantikz}
}
\end{equation}
Alternatively, it is possible to implement each unitary $U_{k,s}$ using a circuit of logarithmic depth and a number of gates linear in the subset size by parallelizing the sequence of CNOT gates~\cite{moore2001parallel}.

Since the Clifford group (up to global phases) is generated by the Hadamard gate $H$, the phase gate $S$, and the controlled-NOT gate $\mathrm{CNOT}$, it follows from the above circuit diagrams that the \NAME unitary ensemble is a Clifford ensemble.
In particular, we can apply all results from Section~\ref{sec:clifford-ensembles} to obtain that the shadow channel of the \NAME ensemble is diagonal in the Pauli basis and given by
\begin{equation}
    \MMM(P) = \lambda_P P, \qquad \lambda_P = \PP_{k, s}[\text{$U_{k, s} P U_{k, s}^\dagger$ is of $Z$-type}].
\end{equation}
If all eigenvalues $\lambda_P$ are non-zero, then the ensemble is tomographically complete and the inverse of the shadow channel is defined by
\begin{equation}
    \MMM^{-1}(P) = \lambda_P^{-1} P.
\end{equation}
In this case, the shadow norm of a Pauli operator $P$ can be expressed in terms of its corresponding eigenvalue $\lambda_P$ as
\begin{equation}
    \norm{P}_{\mathrm{sh}}^2 = \lambda_P^{-1}.
\end{equation}

Next, we present three concrete probability distributions for the \NAME ensemble before we compute the eigenvalues of the shadow channel in Section~\ref{sec:channel-spectrum} and derive bounds on the shadow norm for general Hermitian operators in Section~\ref{sec:shadow-norm}.
In particular, we will show in Section~\ref{sec:channel-spectrum} that if the probability distribution $p$ is strictly positive, i.e.,\ $p(s) > 0$ for all $s \in \{0, 1\}^n$, then the \NAME ensemble is tomographically complete and we can run the classical shadow protocol in Algorithm~\ref{alg:classical-shadow}.

\begin{example}\label{ex:probability-distributions}
The following are three natural examples of probability distributions for the \NAME ensemble that will be used throughout the paper.
\begin{enumerate}[ref=\theexample.\arabic*]
    \item \emph{Uniform distribution}\label{item:uniform-distribution}: Assume every $s \in \{0, 1\}^n$ is chosen with equal probability. Then the probability of each $s \in \{0, 1\}^n$ is given by
    \begin{equation}
        p(s) = \frac{1}{2^n}.
    \end{equation}
    \item \emph{Binomial distribution}\label{item:binomial-distribution}: Let $q \in (0, 1)$ and assume that every bit $s_i$ is chosen independently with probability $q$, i.e.,\
    \begin{equation}
        s_i = \begin{cases}
            1, & \text{with probability $q$}, \\
            0, & \text{with probability $1-q$}.
        \end{cases}
    \end{equation}
     Then the probability of $s \in \{0, 1\}^n$ is given by
    \begin{equation}
        p(s) = q^{\wt(s)} (1-q)^{n-\wt(s)}.
    \end{equation}
    \item \emph{Uniform--uniform distribution}\label{item:uniform-uniform-distribution}: Assume that $m \in \{0, 1, \dots, n\}$ is chosen uniformly at random and then $s \in \{0, 1\}^n$ is chosen uniformly at random among all bitstrings with Hamming weight $\wt(s) = m$.
    Then the probability of $s \in \{0, 1\}^n$ is given by
    \begin{equation}
        p(s) = \frac{1}{n+1} \binom{n}{\wt(s)}^{-1}.
    \end{equation}
\end{enumerate}

\end{example}

\section{Spectrum of the shadow channel}\label{sec:channel-spectrum}

In this section, we compute the eigenvalues of the \NAME shadow channel $\MMM$ for a general probability distribution $p$. Since the \NAME ensemble is a Clifford unitary ensemble, Proposition~\ref{prop:M-eigenvalues} implies that the shadow channel is diagonal in the Pauli basis with eigenvalues given by
\begin{equation}
    \lambda_P = \PP_{k, s}[\text{$U_{k, s} P U_{k, s}^\dagger$ is of $Z$-type}].
\end{equation}
We begin by establishing an alternative characterization of this $Z$-type condition in the following lemmas, before we show in Theorem~\ref{thm:main-channel-eigenvalues} that the eigenvalues of the \NAME shadow channel
can be expressed in terms of the probability distribution $p$ and its Fourier transform $\widehat{p}$ over $\ZZ_2^n$.
To simplify sign computations, we will index Pauli operators by bitstrings $a, b \in \{0, 1\}^n$ and formulate statements using characters $\chi_{a}$ of $\ZZ_2^n$. See Section~\ref{sec:preliminaries} for a detailed introduction to this notation.

\begin{lemma}\label{lem:GHZ-eigen-forall}
Let $a, b \in \{0, 1\}^n$ and $k \in \{0, 1\}$. Then $\ket{\varphi_{k,x}}$ is an eigenvector of $P_{a,b}$ for all $x \in \{0, 1\}^n$ if and only if
\begin{enumerate}
\item $a = \zero$ and $\chi_\one(b) = 1$, or
\item $a = \one$ and $\chi_\one(b) = (-1)^k$.
\end{enumerate}
\end{lemma}
\begin{proof}
Let $x \in \{0, 1\}^n$ and define $y = (0, x_2, \dots, x_n)$. Then
\begin{equation}
    \ket{\varphi_{k,x}} = \frac{1}{\sqrt{2}} \bigl(\ket{y} + \alpha \ket{\overline{y}} \bigr)
    \quad \text{with} \quad
    \alpha = (-1)^{x_1} \ii^k
\end{equation}
and Equation~\eqref{eq:Pauli-basis-action} yields
\begin{equation}
    P_{a,b}\ket{\varphi_{k,x}}
    = \frac{1}{\sqrt{2}} \bigl(\ii^{a \cdot b} \, \chi_b(y)\ket{y+a} + \alpha \, \ii^{a \cdot b} \, \chi_b(\overline{y})\ket{\overline{y}+a} \bigr).
\end{equation}
For $P_{a,b}\ket{\varphi_{k,x}} = \lambda \ket{\varphi_{k,x}}$ to hold, we need $\{ y, \overline{y}  \} = \{ y + a, \overline{y} + a \}$, which is precisely the case if $a = \zero$ or $a = \one$. First, assume $a = \zero$. Then $\ii^{a \cdot b} = 1$ and comparing coefficients yields
\begin{equation}
    \lambda = \chi_b(y), \qquad \lambda = \chi_b(\overline{y}) = \chi_\one(b) \chi_b(y).
\end{equation}
Dividing both equations by $\chi_b(y)$ gives $\chi_\one(b) = 1$. Second, assume $a = \one$. Then we have
\begin{align}
    \lambda = \alpha \, \ii^{a \cdot b} \,  \chi_b(\overline{y}) = \alpha \, \ii^{a \cdot b} \, \chi_\one(b) \chi_b(y),
    \qquad
    \alpha \lambda = \ii^{a \cdot b} \, \chi_b(y).
\end{align}
Multiplying the first equation by $\alpha$
and dividing both equations by $\ii^{a \cdot b} \, \chi_b(y)$ yields
 $\alpha^2 \chi_\one(b) = 1$. By definition,
\begin{equation}
    \alpha^2 = (-1)^{2x_1} \, \ii^{2k} = (-1)^k.
\end{equation}
Thus, we require $\chi_\one(b) = (-1)^k$.
\end{proof}

\begin{lemma}\label{lem:UPU-Z-type}
Let $a, b \in \{0, 1\}^n$, $k \in \{0, 1\}$ and $s \in \{0, 1\}^n$. Then $U_{k,s} P_{a,b} U_{k,s}^\dagger$ is of $Z$-type if and only if
\begin{enumerate}
    \item $a = \zero$ and $\chi_b(s) = 1$, or
    \item $a = s$ and $\chi_b(s) = (-1)^k$.
\end{enumerate}
\end{lemma}
\begin{proof}
Let $I = \{ i : s_i = 1 \} \subseteq \{1, \dots, n\}$ and $m = |I|$. For a subset $J = \{ j_1 < \dots < j_\ell\}$ and $x \in \{0, 1\}^n$, write $x_J = (x_{j_1}, \dots, x_{j_\ell})$.
It follows from Proposition~\ref{prop:UPU-alt} that $U_{k,s} P_{a,b} U_{k,s}^\dagger$ is of $Z$-type if and only if $U_{k,s}^\dagger\ket{x}$ is an eigenvector of $P_{a,b}$ for all $x \in \{0, 1\}^n$. Since $U_{k,s}$ transforms the qubits in $I$ from the $\GGG_{k,m}$ basis into the computational basis, this is equivalent to
\begin{enumerate}
    \item $\ket{\varphi_{k,x}}$ is an eigenvector of $P_{a_I, b_I}$
    for all $x \in \{0, 1\}^{|I|}$, and
    \item $\ket{x}$ is an eigenvector of $P_{a_{I^c}, b_{I^c}}$ for all $x \in \{0, 1\}^{n - |I|}$.
\end{enumerate}
By Lemma~\ref{lem:GHZ-eigen-forall}, the first condition is equivalent to either $a_I = \zero_I$ and $\chi_\one(b_I) = 1$, or $a_I = \one_I$ and $\chi_\one(b_I) = (-1)^k$. Furthermore, the second condition is equivalent to $a_{I^c} = \zero_{I^c}$. Hence, combining these conditions yields
$a = \zero$ and $\chi_b(s) = 1$, or $a = s$ and $\chi_b(s) = (-1)^k$.
\end{proof}

Using the previous lemma, we can now apply Proposition~\ref{prop:M-eigenvalues} to explicitly compute the eigenvalues of the \NAME shadow channel for a general probability distribution $p$.

\begin{theorem}\label{thm:main-channel-eigenvalues}
Let $a, b \in \{0, 1\}^n$. Then the eigenvalue corresponding to $P_{a,b}$ for the \NAME shadow channel with arbitrary distribution $p$ is
\begin{equation}
    \lambda_{a,b} = \begin{cases}
        \frac{1}{2} + \frac{1}{2} \widehat{p}(b), & \text{if $a = \zero$}, \\
        \frac{1}{2} p\bigl(a\bigr), & \text{if $a \neq \zero$},
    \end{cases}
\end{equation}
where
\begin{equation}
    \widehat{p}(b) = \sum_{s \in \{0, 1\}^n} p(s) \chi_b(s)
\end{equation}
is the Fourier transform of $p$ over $\ZZ_2^n$.
\end{theorem}
\begin{proof}
By Proposition~\ref{prop:M-eigenvalues}, the eigenvalue corresponding to $P_{a,b}$ for the channel is given by
\begin{equation}
    \lambda_{a,b} = \PP_{k \in \{0, 1\}, s \in \{0, 1\}^n}[\text{$U_{k,s} P_{a,b} U_{k,s}^\dagger$ is of $Z$-type}].
\end{equation}
First, assume $a = \zero$. Then both conditions in Lemma~\ref{lem:UPU-Z-type} collapse to $\chi_b(s) = 1$ since $\chi_b(\zero) = 1$, and we have
\begin{equation}
    \lambda_{a,b} = \PP_{k \in \{0,1\}, s \in \{0, 1\}^n}[\chi_b(s) = 1] = \PP_{s \in \{0, 1\}^n}[\chi_b(s) = 1].
\end{equation}
Expanding the probability yields
\begin{equation}
    \lambda_{a,b} = \sum_{s \in \{0, 1\}^n} p(s) \cdot \frac{1}{2}\bigl(\chi_b(s) + 1\bigr) = \frac{1}{2}\widehat{p}(b) + \frac{1}{2}.
\end{equation}
Second, assume $a \neq \zero$. Then both conditions in Lemma~\ref{lem:UPU-Z-type} collapse to $a = s$ and $\chi_b(s) = (-1)^k$. Therefore,
\begin{equation}
    \lambda_{a,b} = \PP_{k \in \{0,1\}, s \in \{0, 1\}^n}[\text{$a = s$ and $\chi_b(a) = (-1)^k$}] = \frac{1}{2}p(a),
\end{equation}
since exactly one unitary $U_{k,s}$ satisfies this condition.
\end{proof}

The previous theorem provides an explicit formula for the eigenvalues of the \NAME shadow channel in terms of the probability distribution $p$ and its Fourier transform $\widehat{p}$.
However, if one is only interested in a lower bound on the eigenvalues, or equivalently an upper bound on the shadow norm of Pauli operators, then the case distinction in the previous theorem can be simplified.

\begin{corollary}\label{corr:general-eigenvalue-bound}
Let $a, b \in \{0, 1\}^n$. Then the eigenvalue corresponding to $P_{a,b}$ for the \NAME shadow channel with arbitrary distribution $p$ is bounded from below by
\begin{equation}
    \lambda_{a,b} \geq \frac{1}{2} p(a).
\end{equation}
\end{corollary}
\begin{proof}
Assume $a = \zero$. Then Theorem~\ref{thm:main-channel-eigenvalues} implies
\begin{equation}
    \lambda_{a,b} = \frac{1}{2} + \frac{1}{2} \widehat{p}(b) = \sum_{s \in \{0, 1\}^n} p(s) \cdot \frac{1}{2}\bigl(\chi_b(s) + 1\bigr).
\end{equation}
Since each summand is non-negative, it follows that
\begin{equation}
    \lambda_{a,b} \geq p(\zero) \cdot \frac{1}{2}\bigl(\chi_b(\zero) + 1\bigr) = p(\zero) \geq \frac{1}{2} p(\zero).
\end{equation}
Assume $a \neq \zero$. Then Theorem~\ref{thm:main-channel-eigenvalues} immediately yields $\lambda_{a,b} = \frac{1}{2} p(a)$.
\end{proof}

As a direct consequence of the previous corollary, we obtain that if the probability distribution $p$ is strictly positive, then all eigenvalues of $\MMM$ are also strictly positive and the inverse $\MMM^{-1}$ exists.
In this case, the \NAME shadow channel is tomographically complete. Furthermore, Corollary~\ref{corr:general-eigenvalue-bound} yields the simple upper bound on the shadow norm of Pauli operators
\begin{equation}
    \norm{P_{a,b}}_{\mathrm{sh}}^2 \leq \frac{2}{p(a)},
\end{equation}
which is independent of $b$ and sharp for $a \neq \zero$.

While the previous corollary provides a simple lower bound on the eigenvalues of the \NAME shadow channel, we require the exact eigenvalues to compute the inverse $\MMM^{-1}$ of the shadow channel, since it is given by
\begin{equation}
    \MMM^{-1}(P_{a,b}) = \lambda_{a,b}^{-1} P_{a,b}.
\end{equation}

To determine the eigenvalues in concrete examples, we can use the fact that the Fourier transform over $\ZZ_2^n$ agrees with the usual Hadamard transform from quantum information up to normalization. Let $p \colon \{0, 1\}^n \to [0, 1]$ be a probability distribution and define the vector
\begin{equation}
    \ket{p} = \sum_{x \in \{0, 1\}^n} p(x) \ket{x}.
\end{equation}
Then the vector $\ket{\widehat{p}}$ corresponding to the Fourier transform of $p$ is given by
\begin{equation}
    \lvert\widehat{p}\rangle = \sqrt{2^n} H^{\otimes n} \ket{p},
\end{equation}
where $H = \frac{1}{\sqrt{2}} \mleft(\begin{smallmatrix}
1 & 1 \\ 1 & -1
\end{smallmatrix}\mright)$ denotes the Hadamard gate. This formulation allows us to derive explicit formulas for the eigenvalues of the shadow channels in Example~\ref{ex:eigenvalues-uniform} and Example~\ref{ex:eigenvalues-binomial}. Moreover, it implies that the Fourier transform of $p$ can always be computed numerically in time $\OOO(n 2^n)$ using a fast Hadamard transform.

\subsection{Example distributions}\label{sec:example-distributions}
In the following, we apply Theorem~\ref{thm:main-channel-eigenvalues} to the three probability distributions from Example~\ref{ex:probability-distributions}
and derive closed formulas for the corresponding eigenvalues of the \NAME shadow channel.

\begin{example}[Eigenvalues for uniform distribution]\label{ex:eigenvalues-uniform}
Assume every $s \in \{0, 1\}^n$ is chosen uniformly with probability $p(s) = \frac{1}{2^n}$ as in Example~\ref{item:uniform-distribution}.
As discussed above, the vector $\ket{\widehat{p}}$ corresponding to the Fourier transform of $p$ is given by
\begin{equation}
    \ket{\widehat{p}} = \sqrt{2^n} H^{\otimes n} \ket{p} = H^{\otimes n} \mleft[\frac{1}{\sqrt{2^n}} \sum_{x \in \{0,1\}^n} \ket{x} \mright]
    = H^{\otimes n} H^{\otimes n} \ket{\zero} = \ket{\zero}.
\end{equation}
Therefore, we immediately have
\begin{equation}
    \widehat{p}(b) = \begin{cases}
        1, & \text{if $b = \zero$}, \\
        0, & \text{if $b \neq \zero$},
    \end{cases}
\end{equation}
and Theorem~\ref{thm:main-channel-eigenvalues} yields
\begin{equation}
    \lambda_{a,b} = \begin{cases}
        1, & \text{if $a = \zero$ and $b = \zero$}, \\
        \frac{1}{2}, & \text{if $a = \zero$ and $b \neq \zero$}, \\
        \frac{1}{2^{n+1}}, & \text{if $a \neq \zero$}.
    \end{cases}
\end{equation}
\end{example}

\begin{example}[Eigenvalues for binomial distribution]\label{ex:eigenvalues-binomial}
\label{ex:binomial-distribution}
Let $q \in (0, 1)$ and assume that every $s \in \{0, 1\}^n$ is chosen as in Example~\ref{item:binomial-distribution} with probability
\begin{equation}
    p(s) = q^{\wt(s)} (1-q)^{n - \wt(s)}.
\end{equation}
Using the previous notation, the vector $\ket{p}$ corresponding to the distribution $p$ can be written as a tensor product
\begin{equation}
    \ket{p} = \bigotimes_{i=1}^n \bigl((1-q) \ket{0} + q \ket{1}\bigr).
\end{equation}
Then its Fourier transform $\ket{\widehat{p}}$ factors as
\begin{equation}
    \ket{\widehat{p}} = \sqrt{2^n} H^{\otimes n} \ket{p} = \bigotimes_{i=1}^n \sqrt{2} H \bigl((1-q) \ket{0} + q \ket{1}\bigr).
\end{equation}
On each single qubit, we compute
\begin{equation}
    \begin{pmatrix}
        1 & 1 \\ 1 & -1
    \end{pmatrix}
    \begin{pmatrix}
        1-q \\ q
    \end{pmatrix}
    = \begin{pmatrix}
        1 \\ 1 - 2q
    \end{pmatrix}.
\end{equation}
Therefore, we have
\begin{equation}
   \widehat{p}(b) = 1^{n - \wt(b)}  \cdot (1 - 2q)^{\wt(b)} = (1 - 2q)^{\wt(b)},
\end{equation}
and Theorem~\ref{thm:main-channel-eigenvalues} yields
\begin{equation}
    \lambda_{a,b} = \begin{cases}
        \frac{1}{2} + \frac{1}{2} (1 - 2q)^{\wt(b)}, & \text{if $a = \zero$}, \\
        \frac{1}{2} q^{\wt(a)} (1-q)^{n - \wt(a)}, & \text{if $a \neq \zero$}.
    \end{cases}
\end{equation}
\end{example}

\begin{example}[Eigenvalues for uniform--uniform distribution]
Assume that every $s \in \{0, 1\}^n$ is chosen as in Example~\ref{item:uniform-uniform-distribution} with probability
\begin{equation}
    p(s) = \frac{1}{n+1} \cdot \binom{n}{\wt(s)}^{-1}.
\end{equation}
Furthermore, assume $a = \zero$. Then Theorem~\ref{thm:main-channel-eigenvalues} implies
\begin{equation}
    \lambda_{\zero,b} = \frac{1}{2} + \frac{1}{2} \widehat{p}(b) = \sum_{x \in \{0, 1\}^n} p(x) \cdot \frac{1}{2}\bigl(\chi_b(x) + 1\bigr)
    = \sum_{\substack{x \in \{0, 1\}^n \\ \text{$x \cdot b$ even}}} p(x).
\end{equation}
Denote by $k = \wt(x)$ and $m = \wt(b)$ the Hamming weights of $x$ and $b$, respectively. Assume that $x \cdot b$ is even and that there are $2\ell$ indices with $x_i = b_i = 1$. Then $x_i = 1$ for $2\ell$ of the $m$ indices with $b_i = 1$ and $x_i = 1$ for $k - 2\ell$ of the $n - m$ indices with $b_i = 0$. Thus,
\begin{equation}
    \lambda_{\zero,b}
    = \sum_{\ell = 0}^{\lfloor m /2 \rfloor} \sum_{k=2\ell}^{n-m+2\ell} \binom{m}{2\ell} \binom{n - m}{k - 2\ell} \cdot \frac{1}{n+1} \binom{n}{k}^{-1}.
\end{equation}
Using the beta function identity
\begin{equation}
    \frac{1}{n+1} \binom{n}{k}^{-1} = B(k + 1, n - k + 1) = \int_0^1 t^k (1 - t)^{n-k} \, \dd t,
\end{equation}
we obtain
\begin{equation}
\lambda_{\zero,b} = \int_0^1 \sum_{\ell = 0}^{\lfloor m /2 \rfloor} \binom{m}{2\ell} t^{2\ell} (1 - t)^{m - 2\ell} \sum_{k=2\ell}^{n-m+2\ell} \binom{n - m}{k - 2\ell} t^{k - 2\ell} (1 - t)^{n-m-(k-2\ell)} \, \dd t.
\end{equation}
By the binomial theorem, we have
\begin{equation}
    \sum_{k = 2\ell}^{n-m+2\ell} \binom{n - m}{k - 2\ell} t^{k - 2\ell} (1 - t)^{n-m-(k-2\ell)}
    = \bigl(t + (1-t)\bigr)^{n - m} = 1.
\end{equation}
Furthermore,
\begin{equation}
    \sum_{\ell=0}^{\lfloor m /2 \rfloor} \binom{m}{2\ell} t^{2\ell} (1 - t)^{m - 2\ell}
    = \frac{1 + (1 - 2t)^{m}}{2}
\end{equation}
is the probability that a $\operatorname{Binomial}(m, t)$ random variable is even. Hence,
\begin{equation}
    \lambda_{\zero,b} = \int_0^1 \frac{1 + (1 - 2t)^{m}}{2} \, \dd t =
     \frac{1}{2} + \frac{1}{2} \mleft[ \frac{(1 - 2t)^{m+1}}{-2(m+1)} \mright]_0^1
     = \frac{1}{2} + \frac{(-1)^{m} + 1}{4(m+1)}
\end{equation}
and Theorem~\ref{thm:main-channel-eigenvalues} yields
\begin{equation}
    \lambda_{a,b} = \begin{cases}
        \frac{1}{2} + \frac{1}{2(\wt(b)+1)}, & \text{if $a = \zero$ and $\wt(b)$ is even}, \\
        \frac{1}{2}, & \text{if $a = \zero$ and $\wt(b)$ is odd}, \\
        \frac{1}{2(n+1)} \binom{n}{\wt(a)}^{-1}, & \text{if $a \neq \zero$}.
    \end{cases}
\end{equation}
In particular, this shows that the Fourier coefficients of $p$ are given by
\begin{equation}
    \widehat{p}(b) = \begin{cases}
        \frac{1}{\wt(b) + 1}, & \text{if $\wt(b)$ is even}, \\
        0, & \text{if $\wt(b)$ is odd}.
    \end{cases}
\end{equation}
\end{example}

\section{Shadow-norm bounds for structured observables}\label{sec:shadow-norm}

In this section, we generalize the following shadow-norm bound for Pauli operators from Section~\ref{sec:channel-spectrum}
to general Hermitian operators:
\begin{equation}
    \norm{P_{a,b}}_{\mathrm{sh}}^2 \leq \frac{2}{p(a)}.
\end{equation}
To formulate this result, we first introduce the following notion of sparsity.

\begin{definition}[$S$-sparsity]\label{def:S-sparse}
Let $S \subseteq \{0,1\}^n$. Then an operator $A$ is \emph{$S$-sparse}
if $a \notin S$ implies $\Tr(AP_{a,b}) = 0$
for all $b \in \{0,1\}^n$.
\end{definition}

Equivalently, an operator $A$ is $S$-sparse if it is supported only on Pauli operators $P_{a,b}$ with $a \in S$. For example, $S = \{\zero\}$ corresponds precisely to operators that are diagonal in the computational basis, while $S = \{ \one \}$ corresponds precisely to anti-diagonal operators.
Figure~\ref{fig:sparsity-structure} illustrates the sparsity structure of a general $2$-qubit operator $A = (a_{ij})_{i,j=1}^4$ by coloring
the matrix entries according to the sparsity pattern. For example, we immediately see that $\{ 10, 11 \}$-sparse operators are block matrices with zero diagonal blocks, while $\{ 00, 01 \}$-sparse operators are block matrices with zero anti-diagonal blocks.
Furthermore, this partitioning agrees precisely with the partitioning of the density matrix $\rho$ induced by the \NAME measurement settings in~\cite{patel2026selective}.

\begin{figure}[htbp]
  \centering
  \begin{tikzpicture}[x=9mm,y=-9mm,line width=0.45pt]
    \foreach \row/\colors in {
      0/{dm00,dm01,dm10,dm11},
      1/{dm01,dm00,dm11,dm10},
      2/{dm10,dm11,dm00,dm01},
      3/{dm11,dm10,dm01,dm00}%
    }{
      \foreach \cellcolor [count=\col from 0] in \colors {
        \fill[\cellcolor]
          (\col,\row) rectangle ++(1,1);
            \pgfmathtruncatemacro{\rowindex}{\row+1}
            \pgfmathtruncatemacro{\colindex}{\col+1}
            \node[font=\Large]
            at ({\col+0.5},{\row+0.5})
            {$a_{\rowindex\colindex}$};
      }
    }
    \foreach \k in {0,...,4} {
      \draw (\k,0) -- (\k,4);
      \draw (0,\k) -- (4,\k);
    }
    \foreach \bits [count=\row from 0] in {00,01,10,11} {
      \draw[fill=dm\bits]
        (5,{\row+0.28}) rectangle ++(0.44,0.44);
      \node[anchor=west,font=\large,inner sep=0pt]
        at (5.65,{\row+0.5}) {$\bits$};
    }
  \end{tikzpicture}
  \caption{Sparsity pattern for a general $2$-qubit operator $A$.}
  \label{fig:sparsity-structure}
\end{figure}

Using the notion of $S$-sparsity, we can now state the main theorem of this section, which allows us to bound the shadow norm of a $S$-sparse Hermitian operator $A$ in terms of the probability distribution $p$ and the operator norm $\norm{A}_\infty$.

\begin{theorem}\label{thm:main-shadow-bound}
Let $S \subseteq \{0,1\}^n$ be non-empty and $A$ be an $S$-sparse Hermitian operator. Then the shadow norm of the \NAME ensemble with a strictly positive distribution $p$ is bounded by
\begin{equation}
    \norm{A}_{\mathrm{sh}}^2
    \leq
    (6 + \pi) \mleft[\max_{s \in S} \frac{1}{p(s)}\mright] \norm{A}^2_\infty.
\end{equation}
\end{theorem}

The proof of Theorem~\ref{thm:main-shadow-bound} is non-trivial and uses both Fourier analysis over $\ZZ_2^n$ and the theory of Markov semigroups applied to heat semigroups on graphs. Hence, we postpone the proof to Section~\ref{sec:proof} and first apply Theorem~\ref{thm:main-shadow-bound} to several examples, where we can compare our bound to existing shadow-norm bounds.
We refer to Section~\ref{sec:experiments} for experimental results and a more practical evaluation of our shadow protocol.

\begin{example}\label{ex:uniform-shadow-norm}
In the simplest case, we can choose $S = \{0, 1\}^n$ and assume that every $s \in S$ is chosen uniformly with probability $p(s) = \frac{1}{2^n}$ as in Example~\ref{item:uniform-distribution}.
Then Theorem~\ref{thm:main-shadow-bound} implies that the shadow norm of an arbitrary Hermitian operator $A$ is bounded by
\begin{equation}
    \norm{A}_{\mathrm{sh}}^2
    \leq
    (6 + \pi) \cdot 2^n \cdot \norm{A}_\infty^2.
\end{equation}
Denote by $\norm{\cdot}_{\mathrm{sh,C}}$ the shadow norm of the global Clifford ensemble in~\cite{huang2020predicting}. Then the authors prove the general bound
\begin{equation}
    \norm{A}_{\mathrm{sh,C}}^2
    \leq
    3 \cdot \Tr(A^2).
\end{equation}
Since $\Tr(A^2) \leq 2^{n} \cdot \norm{A}_\infty^2$, with equality for example for Pauli operators, it follows that
\begin{equation}
    \norm{A}_{\mathrm{sh,C}}^2
    \leq
    3 \cdot 2^n \cdot \norm{A}_\infty^2.
\end{equation}
Hence, our bound for the \NAME shadow norm in the case of the uniform distribution differs only by a constant factor of approximately $3$ from the corresponding bound for the global Clifford shadow norm.
According to~\cite[Theorem 2]{huang2020predicting}, the $2^n$ scaling in $n$ is optimal in the worst case for protocols based on single-copy measurements and arbitrary observables.
\end{example}

\begin{example}\label{ex:uniform-on-S}
Let $S \subseteq \{0, 1\}^n$ be a proper non-empty subset. Fix $\delta \in (0, 1)$ and consider the probability distribution
\begin{equation}
    p(s) = \begin{cases}
        \frac{\delta}{|S|}, & \text{if $s \in S$}, \\
        \frac{1 - \delta}{2^n - |S|}, & \text{otherwise}.
    \end{cases}
\end{equation}
As $\delta \to 1$, this distribution minimizes the bound in Theorem~\ref{thm:main-shadow-bound} by approaching the uniform distribution over $S$ while guaranteeing the invertibility of the shadow channel for $\delta < 1$. In this case, we have
\begin{equation}
    \norm{A}_{\mathrm{sh}}^2
    \leq
    (6 + \pi) \cdot \frac{|S|}{\delta} \cdot \norm{A}_\infty^2.
\end{equation}
Following the discussion in Example~\ref{ex:uniform-shadow-norm}, the corresponding bound for the global Clifford ensemble is
\begin{equation}
    \norm{A}_{\mathrm{sh,C}}^2
    \leq
    3 \cdot 2^n \cdot \norm{A}_\infty^2
\end{equation}
for $S$-sparse operators, since $\Tr(A^2) = 2^{n} \cdot \norm{A}_\infty^2$ is attained by Pauli operators $P_{s,b}$ with $s \in S$.
Therefore, we improve the shadow-norm bound for the global Clifford ensemble by a factor of
\begin{equation}
    \frac{6 + \pi}{3 \delta} \cdot \frac{|S|}{2^n}
\end{equation}
if $|S| < \frac{3 \delta}{6 + \pi} \cdot 2^n$. In particular, this improvement can be exponential if $|S|$ is exponentially smaller than $2^n$.
\end{example}

\subsection{Shadow-norm bounds for \texorpdfstring{$X/Y$}{X/Y}-local observables}\label{sec:XY-local-observables}

We now consider a natural family of structured observables for which the uniform--uniform distribution and the binomial distribution in Example~\ref{ex:binomial-distribution} can be used to significantly reduce the shadow-norm bound compared to both the bound for the \NAME ensemble with a uniform distribution, as well as the bounds for the local Pauli and global Clifford ensembles in~\cite{huang2020predicting}.

Let $k \in \NN$. Then we call an operator \emph{$k$-$X/Y$-local} if it is supported only on Pauli operators with at most $k$ tensor factors belonging to $\{X, Y\}$. Equivalently, an operator is $k$-$X/Y$-local if it is $S$-sparse for
\begin{equation}
    S = \bigl\{ s \in \{0, 1\}^n : \wt(s) \leq k \bigr\}.
\end{equation}

As a simple example, consider the following operator on $4$ qubits which
is $2$-$X/Y$-local since each Pauli operator has at most $2$ tensor factors belonging to $\{X, Y\}$:
\begin{equation}
    X \otimes I \otimes Z \otimes Y + Z \otimes Z \otimes X \otimes I.
\end{equation}
It follows directly from the definition of $X/Y$-locality that every $k$-local operator is also $k$-$X/Y$-local. However, the converse does not hold. For example,
the operator $Z \otimes \dots \otimes Z$ is $0$-$X/Y$-local but $n$-local.
Besides local operators, other physically relevant examples of $X/Y$-local operators arise in the study of fermionic modes under the Jordan--Wigner transformation~\cite{zhao2021fermionic} or of long-range Kitaev chains~\cite{vodola2014kitaev}.
See Section~\ref{sec:experiments} for more details on the first case and for a practical discussion.

Since the $X/Y$-locality is a special instance of $S$-sparsity, the bound in Theorem~\ref{thm:main-shadow-bound} can be reformulated for this class of observables as
\begin{equation}
    \norm{A}_{\mathrm{sh}}^2
    \leq
    (6 + \pi) \mleft[\max_{\wt(s) \leq k} \frac{1}{p(s)}\mright] \norm{A}^2_\infty.
\end{equation}

In the following examples, we show that, for suitable probability distributions and fixed $k$, we can obtain a bound that is polynomial in $n$. Furthermore, we compare this bound to the corresponding bounds for the local Pauli ensemble and global Clifford ensemble from~\cite{huang2020predicting} at the end of this section.

\begin{example}
Consider the probability distribution from Example~\ref{item:uniform-uniform-distribution} where we first choose a Hamming weight $m \in \{0, \dots, n\}$ uniformly at random and then choose a bitstring $s \in \{0, 1\}^n$ uniformly at random among all bitstrings with $\wt(s) = m$. Using the loose bound $\binom{n}{m} \leq n^m$, we obtain
\begin{equation}
    \frac{1}{p(s)} = (n + 1) \cdot \binom{n}{m} \leq (n + 1) \cdot n^{m} \leq 2 n^{m + 1}.
\end{equation}
Hence, Theorem~\ref{thm:main-shadow-bound} implies that the shadow norm of a $k$-$X/Y$-local Hermitian operator $A$ is bounded by
\begin{equation}
    \norm{A}_{\mathrm{sh}}^2
    \leq (6 + \pi) \cdot 2 n^{k + 1} \cdot \norm{A}_\infty^2.
\end{equation}
For fixed $k$, this bound is polynomial in $n$ and asymptotically improves the exponential bounds for the uniform \NAME ensemble and the global Clifford ensemble discussed in Example~\ref{ex:uniform-shadow-norm}.
\end{example}

The next example shows that the previous exponent of $k + 1$ can be improved to
$k$ by using the binomial distribution in Example~\ref{ex:binomial-distribution} with the right choice of parameter.

\begin{example}\label{ex:binomial-shadow-norm}

Consider the binomial distribution from Example~\ref{ex:binomial-distribution} with parameter $q = \frac{1}{n + 1}$. Then $1 - q = \frac{n}{n+1}$ and we have
\begin{equation}
    \frac{1}{p(s)} = (n + 1)^{{\wt(s)}} \mleft(\frac{n+1}{n}\mright)^{n-\wt(s)}
    =
    \mleft(\frac{n + 1}{n}\mright)^n n^{\wt(s)}.
\end{equation}
Since
\begin{equation}
    \mleft(\frac{n + 1}{n}\mright)^n
    =
    \mleft(1 + \frac{1}{n}\mright)^n
    \leq e,
\end{equation}
it follows that $\frac{1}{p(s)} \leq e \cdot n^{\wt(s)}$.
Therefore, Theorem~\ref{thm:main-shadow-bound} implies that the shadow norm of a $k$-$X/Y$-local Hermitian operator $A$ is bounded by
\begin{equation}
    \norm{A}_{\mathrm{sh}}^2
    \leq (6 + \pi) \cdot e \cdot n^k \cdot \norm{A}_\infty^2.
\end{equation}
This improves the previous bound for the uniform--uniform distribution asymptotically by a factor of $n$.
\end{example}

The choice of parameter $q = \frac{1}{n+1}$ in the previous example
does not depend on the locality parameter $k$ and yields the asymptotic scaling $\OOO(n^k)$ without knowing the $X/Y$-locality of the observable in advance.
However, if we know an upper bound on the $X/Y$-locality parameter $k$ in advance, we can adjust the parameter $q$ to further improve the constant factor.

\begin{example}\label{ex:binomial-tuned}
Consider the binomial distribution from Example~\ref{ex:binomial-distribution}, and let $s \in \{0, 1\}^n$ with Hamming weight $\wt(s) = m$. Then we have
\begin{equation}
    \frac{1}{p(s)} = q^{-m} (1-q)^{-(n-m)}.
\end{equation}
Minimizing this expression with respect to $q$ yields the optimal choice $q = \frac{m}{n}$, which gives
\begin{equation}
    \frac{1}{p(s)} = 2^{n H_2(m/n)},
    \qquad
    H_2(t) = -t \log_2(t) - (1-t) \log_2(1-t),
\end{equation}
where $H_2$ denotes the binary entropy function, and we use the convention $0 \cdot \log_2(0) = 0$. Now consider a $k$-$X/Y$-local Hermitian operator $A$ with $1 \leq k \leq \frac{n}{2}$
and choose $q = \frac{k}{n}$. Since $q \leq \frac{1}{2}$, we have for every $0 \leq m \leq k$ that
\begin{equation}
    q^{-m}(1-q)^{-(n-m)}
    \leq
    q^{-k}(1-q)^{-(n-k)}
    =
    2^{nH_2(k/n)}.
\end{equation}
Therefore, Theorem~\ref{thm:main-shadow-bound} implies that the shadow norm of $A$ is bounded by
\begin{equation}
    \norm{A}_{\mathrm{sh}}^2
    \leq (6 + \pi) \cdot 2^{n H_2(k/n)} \cdot \norm{A}_\infty^2.
\end{equation}
Using the estimate $H_2(t) \leq t \log_2(e/t)$, we obtain
\begin{equation}
    2^{n H_2(k/n)}
    \leq 2^{k \log_2(en/k)}
    = \mleft(\frac{en}{k}\mright)^{k} = \mleft(\frac{e}{k}\mright)^{k} \cdot n^k.
\end{equation}
This bound has the same asymptotic scaling as in the previous example, but improves the constant factor if $k \geq 2$. Moreover, if $k$ is not fixed but of order $\Theta(n)$, then this bound has a scaling of $2^{\OOO(n)}$ instead of $n^{\OOO(n)}$ as in Example~\ref{ex:binomial-shadow-norm}.

Using a similar argument, one can derive an analogous bound for the complementary case of observables with all Pauli operators in their support having at least $k$ tensor factors belonging to $\{X, Y\}$; see also Section~\ref{sec:exp-XY}.
\end{example}

The bounds in the above examples are polynomial in $n$ for fixed $k$, which significantly improves the exponential scaling $2^n$ of the global Clifford bound as discussed previously. In~\cite{huang2020predicting}, the authors also study the local Pauli ensemble, which admits a shadow-norm bound of the form
\begin{equation}
    \norm{A}_{\mathrm{sh,P}}^2
    \leq
    4^k \cdot \norm{A}_\infty^2
\end{equation}
for $k$-local operators. Thus, in the case of $k$-local operators, our bounds for the \NAME ensemble have a worse scaling of $n^k$ compared to $4^k$ for the local Pauli ensemble. However, for $k$-$X/Y$-local operators that are not $k$-local, such as $Z \otimes \dots \otimes Z$, our bounds can yield an exponential improvement over the local Pauli ensemble.

\section{Proof of Theorem~\ref{thm:main-shadow-bound}}\label{sec:proof}

In this section, we present a proof of Theorem~\ref{thm:main-shadow-bound}, which states that the \NAME shadow norm of an $S$-sparse Hermitian operator $A$ is bounded by
\begin{equation}
    \norm{A}_{\mathrm{sh}}^2
    \leq
    (6 + \pi) \mleft[\max_{s \in S} \frac{1}{p(s)}\mright] \norm{A}^2_\infty.
\end{equation}
The first step will be a decomposition of the Hermitian operator $A$
into the disjoint Pauli sectors $A_{k,s}$ according to the conditions in Lemma~\ref{lem:UPU-Z-type}:
\begin{equation}
    A = \sum_{k\in\{0,1\}} \sum_{s\in\{0,1\}^n} A_{k,s},
    \qquad
    A_{k,s} = \frac{1}{2^n} \sum_{\substack{b\in\{0,1\}^n\\ \chi_b(s)=(-1)^k}} \Tr(AP_{s,b}) P_{s,b}.
\end{equation}
Then we can estimate the shadow norm on each sector $A_{k,s}$ separately and combine the estimates using the following lemma to obtain a bound on the total shadow norm $\norm{A}_{\mathrm{sh}}^2$ in terms of $\norm{A}_\infty^2$. For the full argument, we refer to the \hyperref[proof:shadow-norm-bound]{proof at the end of this section}.

\begin{lemma}\label{lem:sector-square-sum}
Let $A$ be a Hermitian operator. For every
$k\in\{0,1\}$ and $s\in\{0,1\}^n$, define
\begin{equation}
    A_{k,s}
    =
    \sum_{\substack{
        b\in\{0,1\}^n\\
        \chi_b(s)=(-1)^k
    }}
    \alpha_{s,b}P_{s,b},
    \qquad
    \alpha_{s,b}
    =
    \frac{1}{2^n}\Tr(AP_{s,b}).
\end{equation}
Then each $A_{k,s}$ is Hermitian and
\begin{equation}
    \sum_{k\in\{0,1\}}
    \sum_{s\in\{0,1\}^n}
    A_{k,s}^2
    =
    \diag(A^2).
\end{equation}
\end{lemma}
\begin{proof}
Since $A$ is Hermitian, each of the Pauli coefficients $\alpha_{s,b}$ is real. Thus, $A_{k,s}$ is Hermitian as a real linear combination of Hermitian operators. For $s \in \{0,1\}^n$, define the combined operator
\begin{equation}
    A_s = A_{0,s}+A_{1,s} = \sum_{b\in\{0,1\}^n}
    \alpha_{s,b}P_{s,b}.
\end{equation}
Denote by $\{A,B\} = AB + BA$ the anticommutator of two operators $A$ and $B$. Then
\begin{equation}
    A_{s}^2 = A_{0,s}^2 + A_{1,s}^2 + \{ A_{0,s}, A_{1,s} \},
\end{equation}
where
\begin{equation}
    \{A_{0,s}, A_{1,s}\} =
    \sum_{\substack{
        a\in\{0,1\}^n\\
        \chi_a(s)=1
    }}
    \sum_{\substack{
        b\in\{0,1\}^n\\
        \chi_b(s)=-1
    }}
    \alpha_{s,a}
    \alpha_{s,b}\{P_{s,a}, P_{s,b}\}.
\end{equation}
According to Eq.~\eqref{eq:Pauli-product}, we have
\begin{equation}
    \{P_{s,a}, P_{s,b}\}
    = P_{s,a} P_{s,b} + P_{s,b} P_{s,a} = \ii^{s \cdot a} \ii^{s \cdot b} (-\ii)^{0} \bigl( \chi_a(s) + \chi_b(s) \bigr) P_{\zero, a + b}.
\end{equation}
Hence, $\{A_{0,s}, A_{1,s}\} = 0$ if $\chi_a(s)=1$ and $\chi_b(s)=-1$, which shows that
\begin{equation}
    A_{s}^2 = A_{0,s}^2 + A_{1,s}^2.
\end{equation}
Summing over all $s \in \{0, 1\}^n$ gives
\begin{equation}
    \sum_{k\in\{0,1\}} \sum_{s \in \{0,1\}^n} A_{k,s}^2
    =
    \sum_{s \in \{0,1\}^n} A_s^2.
\end{equation}
Next, we compute $A_s$ in the computational basis. Let $x, y \in \{0, 1\}^n$. Then Eq.~\eqref{eq:Pauli-basis-action} yields
\begin{equation}
    \bra{x} P_{s,b} \ket{y} = \ii^{s \cdot b} \chi_b(y) \braket{x}{y + s} = \begin{cases}
        \ii^{s \cdot b} \chi_b(y), & \text{if $y = x + s$},\\
        0, & \text{otherwise},
    \end{cases}
\end{equation}
and
\begin{equation}
    \Tr(AP_{s,b})
    =
    \sum_{x\in\{0,1\}^n}
    \bra{x}AP_{s,b}\ket{x}
    =
    \ii^{s \cdot b}
    \sum_{x\in\{0,1\}^n}
    \chi_b(x) \bra{x}A\ket{x + s}.
\end{equation}
Consider the matrix element
\begin{equation}
    \bra{x}A_s\ket{y}
    =
    \frac{1}{2^n}
    \sum_{b\in\{0,1\}^n}
    \Tr(AP_{s,b})
    \bra{x}P_{s,b}\ket{y}.
\end{equation}
If $y \neq x + s$, then $\bra{x}P_{s,b}\ket{y}=0$ for all $b \in \{0,1\}^n$ and
$\bra{x}A_s\ket{y}=0$. Otherwise, $y= x+s$ and
\begin{equation}\label{eq:As-nonzero-matrix-element}
    \bra{x}A_s\ket{x+s}
    =
    \frac{1}{2^n}
    \sum_{b, z\in\{0,1\}^n}
    \ii^{s \cdot b}
    \ii^{s \cdot b}
    \chi_b(z)
    \chi_b(x + s)
    \bra{z}A\ket{z + s}.
\end{equation}
Note that
\begin{equation}
    \ii^{s \cdot b}
    \ii^{s \cdot b}
    \chi_b(z)
    \chi_b(x + s)
    = \chi_b(s) \chi_{b}(z) \chi_{b}(x) \chi_b(s)
    = \chi_{z}(b) \chi_{x}(b).
\end{equation}
Hence, character orthogonality implies
\begin{equation}
    \frac{1}{2^n}
    \sum_{b \in\{0,1\}^n}
    \ii^{s \cdot b}
    \ii^{s \cdot b}
    \chi_b(z)
    \chi_b(x + s)
    = \begin{cases}
        1, & \text{if $x = z$}, \\
        0, & \text{otherwise},
    \end{cases}
\end{equation}
and Eq.~\eqref{eq:As-nonzero-matrix-element} simplifies to
\begin{equation}
    \bra{x}A_s\ket{x+s} = \bra{x}A\ket{x + s}.
\end{equation}
Therefore, we have
\begin{equation}
    A_s
    =
    \sum_{x \in \{0, 1\}^n}
    \bra{x}A\ket{x+s} \ketbra{x}{x+s}.
\end{equation}
Squaring and using $y = x + s$ in the resulting sum yields
\begin{align}
    A_s^2
    &=
    \sum_{x, y \in \{0, 1\}^n}
    \bra{x}A\ket{x+s} \bra{y}A\ket{y+s} \ketbra{x}{x+s}\ketbra{y}{y+s} \\
    &=
    \sum_{x \in \{0, 1\}^n}
    \bra{x}A\ket{x+s} \bra{x+s}A\ket{x} \ketbra{x}{x}.
\end{align}
Since $A$ is Hermitian,
\begin{equation}
    \bra{x+s}A\ket{x} = \overline{\bra{x}A\ket{x+s}}.
\end{equation}
Thus,
\begin{equation}
    \sum_{s \in \{0, 1\}^n} A_s^2 = \sum_{s,x \in \{0, 1\}^n}
    \lvert\bra{x}A\ket{x+s}\rvert^2 \ketbra{x}{x}.
\end{equation}
For fixed $x$, the map $s \mapsto x+s$ defines a bijection on $\{0, 1\}^n$. Hence,
\begin{equation}
    \sum_{s \in \{0, 1\}^n} \lvert\bra{x}A\ket{x+s}\rvert^2
    =
    \sum_{y \in \{0, 1\}^n} \lvert\bra{x}A\ket{y}\rvert^2.
\end{equation}
Since $A$ is Hermitian, it follows that
\begin{equation}
    \bra{x}A^2\ket{x}
    =
    \sum_{y \in \{0, 1\}^n} \bra{x}A\ket{y} \overline{\bra{x}A\ket{y}}
    =
    \sum_{y \in \{0, 1\}^n} \lvert\bra{x}A\ket{y}\rvert^2.
\end{equation}
Therefore, we conclude that
\begin{equation}
    \sum_{k \in \{0, 1\}} \sum_{s \in \{0, 1\}^n} A_{k,s}^2
    =
    \sum_{x \in \{0, 1\}^n} \bra{x}A^2\ket{x} \ketbra{x}{x}
    =
    \diag(A^2).
\end{equation}
\end{proof}

To estimate the shadow norm on each of the Pauli sectors $A_{k,s}$, we will perform a case distinction.
In the case $s \neq \zero$, we can use Proposition~\ref{prop:M-eigenvalues} and Theorem~\ref{thm:main-channel-eigenvalues} to
explicitly compute and estimate the terms of the form $\MMM_{\ell,r}(\MMM^{-1}(A_{k,s}))^2$ that arise from the definition of the shadow norm.
Additionally, we have $A_{1,\zero} = 0$ in the case $s = \zero$ and $k = 1$. However, the Pauli operators in the remaining case $A_{0,\zero}$ do not share the same eigenvalues in Theorem~\ref{thm:main-channel-eigenvalues}. Therefore, we need a different approach to estimate the shadow norm in this setting.

Since $A_{0,\zero}$ is diagonal in the computational basis, we can restrict ourselves to the subspace of diagonal operators $A$ and identify them with functions
\begin{equation}
    f_A \in \ell^\infty(\ZZ_2^n), \qquad f_A(x) = \bra{x}A\ket{x}.
\end{equation}
Note that under this identification, the operator norm of $A$ agrees with the $\ell^\infty$-norm of the function $f_A$ since
\begin{equation}
    \norm{A}_\infty = \max_{x \in \{0, 1\}^n} |\bra{x}A\ket{x}| = \norm{f_A}_\infty.
\end{equation}
Moreover, in this setting, the auxiliary channels $\MMM_{k,s}$ can be expressed as follows.

\begin{lemma}\label{lem:translation-operator}
Let $k \in \{0, 1\}$ and $s \in \{0, 1\}^n$. After identifying diagonal operators with functions $f \in \ell^\infty(\ZZ_2^n)$,
the channel $\MMM_{k,s}$ satisfies
\begin{equation}
    (\MMM_{k,s} f)(x) = \frac{1}{2}\bigl(f(x) + f(x + s)\bigr).
\end{equation}
\end{lemma}
\begin{proof}
Note that each Pauli operator $P_{\zero, b}$ is diagonal in the computational basis and corresponds to the character $\chi_b \in \ell^\infty(\ZZ_2^n)$ under our identification
because $\bra{x}P_{\zero, b}\ket{x} = \chi_b(x)$. Hence, we can use Proposition~\ref{prop:M-eigenvalues} and Lemma~\ref{lem:UPU-Z-type} to compute
\begin{equation}
    (\MMM_{k,s} \chi_b)(x)
    =
    \bra{x}\MMM_{k,s}(P_{\zero, b})\ket{x}
    = \begin{cases}
        \chi_b(x), & \text{if $\chi_b(s) = 1$},\\
        0, & \text{otherwise}.
    \end{cases}
\end{equation}
Therefore, we have
\begin{equation}
    (\MMM_{k,s} \chi_b)(x) = \frac{1}{2}\bigl(1 + \chi_b(s)\bigr) \cdot \chi_b(x)
    = \frac{1}{2}\bigl(\chi_b(x) + \chi_b(x + s)\bigr).
\end{equation}
The statement follows by linearity since the characters $\chi_b$ form a basis of $\ell^\infty(\ZZ_2^n)$.
\end{proof}

The main insight to establish the shadow-norm bound for $A_{0,\zero}$ is the following observation. Using Lemma~\ref{lem:translation-operator}, the shadow channel $\MMM$ and the corresponding squared expectation can be written as
\begin{equation}
    \EE[\MMM_{k,s}f](x) = \frac{1}{2} \sum_{s \in \{0, 1\}^n} p(s) \bigl(f(x) + f(x + s)\bigr),
\end{equation}
\begin{equation}
    \EE[(\MMM_{k,s}f)^2](x) = \frac{1}{4} \sum_{s \in \{0, 1\}^n} p(s) \bigl(f(x) + f(x + s)\bigr)^2.
\end{equation}
These expressions can be realized as the graph Laplacian and the Dirichlet energy of a finite graph when performing a sign correction. This allows us to apply tools from
the theory of Markov semigroups to the heat semigroup on the graph to obtain the desired bound for the diagonal sector.

The following technical lemma constructs the corresponding graph and establishes the precise bound used in \hyperref[proof:shadow-norm-bound]{proof at the end of this section}. It itself relies on Lemma~\ref{lem:graph-energy-estimate} that will be discussed and proven afterwards. We refer to~\cite{chung1996spectral,levin2017markov} for more information on graph Laplacians and Dirichlet forms in the context of graphs and Markov chains, and to~\cite{zaslavsky1982signed,reiner2014critical} for more details on the graph double-cover construction used in the following proof.

\begin{lemma}\label{lem:square-func-bound}
Let $f \colon \{0,1\}^n \to \RR$ and assume $p(\zero) > 0$. Using the identification of diagonal operators with $\ell^\infty(\ZZ_2^n)$,
define $\LLL f = \EE_{k,s}[(\MMM_{k,s}f)^2]$. Then
\begin{equation}
    \norm{\LLL f}_\infty
    \leq \frac{1 + \pi/2}{p(\zero)} \cdot \norm{\MMM f }^2_\infty.
\end{equation}
\end{lemma}
\begin{proof}
Using Lemma~\ref{lem:translation-operator}, we can expand the shadow channel $\MMM = \EE_{k, s}[\MMM_{k,s}]$ and the
operator $\LLL$ into
\begin{equation}
   (\MMM f)(x) = \frac{1}{2} \sum_{s \in \{0,1\}^n} p(s) \bigl(f(x) + f(x + s)\bigr),
\end{equation}
\begin{equation}
    (\LLL f)(x) = \frac{1}{4} \sum_{s \in \{0,1\}^n} p(s) \bigl(f(x) + f(x + s)\bigr)^2.
\end{equation}
Define the bipartite graph with vertices $V = \ZZ_2^n \times \ZZ_2$ and with non-negative, translation-invariant weights
\begin{equation}
w\bigl((x,\varepsilon),(y,\delta)\bigr)
= \begin{cases}
        p(x + y), & \text{if $x \neq y$ and $\delta \neq \varepsilon$},\\
        0, & \text{otherwise}.
    \end{cases}
\end{equation}
For any function $f \colon \{0,1\}^n \to \RR$, define its extension
\begin{equation}
    \widetilde{f} \colon V \to \RR,
    \quad
    \widetilde{f}(x, \varepsilon) = (-1)^\varepsilon f(x)
\end{equation}
and the graph Laplacian and energy operator
\begin{equation}
(\Delta \widetilde{f})(x, \varepsilon) = \sum_{(y, \delta) \in V} w\bigl((x, \varepsilon), (y, \delta) \bigr)\bigl(\widetilde{f}(x, \varepsilon) - \widetilde{f}(y, \delta)\bigr),
\end{equation}
\begin{equation}
(\EEE\widetilde{f})(x, \varepsilon) = \frac{1}{2} \sum_{(y, \delta) \in V} w\bigl((x, \varepsilon), (y, \delta) \bigr)\bigl(\widetilde{f}(x, \varepsilon) - \widetilde{f}(y, \delta)\bigr)^2.
\end{equation}
Then, we have
\begin{align}
    \Delta(\widetilde{f})(x, \varepsilon) &= \sum_{y \in \ZZ_2^n \setminus \{x\}} p(x + y) \bigl(\widetilde{f}(x,\varepsilon) - \widetilde{f}(y,\varepsilon + 1) \bigr) \\
    &= (-1)^\varepsilon \sum_{y \in \ZZ_2^n \setminus \{\zero\}} p(y) \bigl(f(x) + f(x + y)\bigr),
\end{align}
and similarly
\begin{equation}
   (\EEE\widetilde{f})(x, \varepsilon)  = \frac{1}{2} \sum_{y \in \ZZ_2^n \setminus \{\zero\}} p(y) \bigl(f(x) + f(x + y)\bigr)^2.
\end{equation}
Moreover, this shows
\begin{equation}
    \widetilde{\MMM f} = p(\zero) \widetilde{f} + \frac{1}{2}  \Delta \widetilde{f},
    \qquad
    \LLL f = p(\zero) f^2 + \frac{1}{2} \EEE \widetilde{f}.
\end{equation}
Therefore, it follows that
\begin{equation}
    \norm{\LLL f}_\infty \leq \frac{1}{p(\zero)} \norm{p(\zero) f}^2_\infty + \frac{1}{2} \bigl\lVert\EEE\widetilde{f}\bigr\rVert_\infty.
\end{equation}
To bound the first term, observe that
\begin{equation}
    (\MMM f)(x) = \frac{1 + p(\zero)}{2} f(x) + \frac{1}{2} \sum_{s \neq \zero} p(s) f(x + s).
\end{equation}
Thus, the reverse triangle inequality yields
\begin{equation}
    \norm{\MMM f}_\infty \geq \frac{1 + p(\zero)}{2} \norm{f}_\infty - \frac{1}{2} \sum_{s \neq \zero} p(s) \norm{f}_\infty = p(\zero) \norm{f}_\infty.
\end{equation}
For the second term, we apply Lemma~\ref{lem:graph-energy-estimate} with $\lambda = 2 p(\zero)$ to obtain
\begin{equation}
    \lVert\EEE\widetilde{f}\rVert_\infty \leq
    \frac{\pi}{4 p(\zero)} \Bigl\lVert 2 \Bigl(p(\zero) \widetilde{f} + \frac{1}{2} \Delta \widetilde{f}\Bigr)\Bigr\rVert_\infty^2
    = \frac{\pi}{p(\zero)} \bigl\lVert\widetilde{\MMM f}\bigr\rVert_\infty^2.
\end{equation}
Hence, combining the previous two bounds gives
\begin{equation}
    \norm{\LLL f}_\infty \leq \frac{1}{p(\zero)} \norm{\MMM f}^2_\infty + \frac{\pi}{2p(\zero)} \bigl\lVert\widetilde{\MMM f}\bigr\rVert_\infty^2 = \frac{1 + \pi/2}{p(\zero)} \norm{\MMM f}^2_\infty,
\end{equation}
where we additionally used that $\norm{g}_\infty = \norm{\widetilde{g}}_\infty$ for all $g \colon \{0,1\}^n \to \RR$.
\end{proof}

The following lemma establishes the key estimate in the previous proof in a purely graph-theoretic setting. It relies on the facts from the theory of Markov semigroups introduced in Section~\ref{sec:markov-semigroups}, particularly the integral representation of the resolvent and a reverse Poincaré inequality. See also~\cite{munch2023non-negative,baudoin2016reverse} for further applications of reverse Poincaré inequalities in the context of heat semigroups on graphs.

\begin{lemma}
\label{lem:graph-energy-estimate}
Consider a weighted graph with vertices $V = \ZZ_2^m$ and
non-negative, translation-invariant weights $w \colon V \times V \to \RR$, i.e.,
\begin{equation}
    w(x, y) \geq 0, \qquad
    w(x + z, y + z) = w(x, y)
    \qquad
    \forall x, y, z \in V.
\end{equation}
For $f \colon V\to \RR$, define the graph Laplacian and the local Dirichlet energy by
\begin{equation}
    (\Delta f)(x)
    =
    \sum_{y\in V}
    w(x,y)\bigl(f(x)-f(y)\bigr),
\end{equation}
\begin{equation}
    (\EEE f)(x)
    =
    \frac12
    \sum_{y\in V}
    w(x,y)
    \mleft(f(x)-f(y)\mright)^2.
\end{equation}
Then it holds for every $\lambda>0$ that
\begin{equation}
    \norm{\EEE f}_\infty
    \leq
    \frac{\pi}{2\lambda}
    \norm{\lambda f +\Delta f}_\infty^2.
\end{equation}
\end{lemma}
\begin{proof}
For every $s \in V$, define the weight $a(s) = w(\zero, s)$ and the translation operator
\begin{equation}
    (T_s f)(x)=f(x + s).
\end{equation}
Then the Laplacian and the energy can be written as
\begin{equation}
    \Delta
    =
    \sum_{s\in V}
    a(s)\bigl(\id-T_s\bigr),
    \qquad
    \EEE f
    =
    \frac12
    \sum_{s\in V}
    a(s)
    \bigl(
        (T_s- \id) f
    \bigr)^2.
\end{equation}
Note that the operators $T_s$ commute because $\ZZ_2^m$ is abelian. Moreover, the negative Laplacian $-\Delta$ generates the continuous semigroup
\begin{equation}
    (P_t)_{t \geq 0}, \quad P_t =e^{-t\Delta}.
\end{equation}
Since $(\Delta 1)(x) = 0$ for all $x \in V$, it follows that
\begin{equation}
    P_t 1 = \sum_{k=0}^\infty \frac{1}{k!}(-t\Delta)^k 1 = \id 1 = 1.
\end{equation}
Let $f \geq 0$. Then $a(s)T_s f \geq 0$ for all $s \in V$, which implies
\begin{equation}
    P_t f
    =
    e^{-t \sum_{s\in V}
    a(s)}
    \sum_{k=0}^\infty \frac{t^k}{k!}\Bigl(\sum_{s \in V} a(s) T_s\Bigr)^k f \geq 0.
\end{equation}
Consequently, each $P_t$ is a Markov operator and $(P_t)_{t \geq 0}$ is a Markov semigroup. Additionally,
\begin{align*}
    \frac{1}{2} \bigl( (-\Delta) f^2 - 2 f (-\Delta) f \bigr)
    &= \frac{1}{2} \sum_{s\in V}
    a(s)\bigl(-f^2 + T_s f^2 + 2 f^2 - 2 f T_s f\bigr) \\
    &= \frac{1}{2} \sum_{s\in V}
    a(s)\bigl(T_s f - f\bigr)^2 \\
    &= \Gamma f.
\end{align*}
Thus, $\EEE =\Gamma$ agrees with the carré du champ operator. Since every $T_s$ commutes with $\Delta$, it also commutes with
$P_t$. Hence,
\begin{align}
    \Gamma(P_tf)(x)
    &=
    \frac12
    \sum_{s\in V}
    a(s)
    \bigl(
        P_t(T_sf-f)(x)
    \bigr)^2
    \\
    &\leq
    \frac12
    \sum_{s\in V}
    a(s)
    P_t\mleft(
        (T_sf-f)^2
    \mright)(x)
    \\
    &=
    P_t\bigl(\Gamma(f)\bigr)(x).
\end{align}
Let $g = \bigl(\lambda\id + \Delta\bigr)f$. Then Proposition~\ref{prop:resolvent-integral} gives
\begin{equation}
    f
    =
    \int_0^\infty e^{-\lambda t}P_tg\,\dd t.
\end{equation}
For $x\in V$, Minkowski's integral inequality implies
\begin{align}
    \sqrt{(\Gamma f)(x)}
    &= \mleft[
    \sum_{s\in V}
    \frac{a(s)}{2}
    \mleft(
        \int_0^\infty e^{-\lambda t} \bigl((T_s - \id)P_tg\bigr)(x)\,\dd t
    \mright)^2 \mright]^{\frac12} \\
    &\leq
    \int_0^\infty
    e^{-\lambda t}
    \mleft[
    \sum_{s\in V}
    \frac{a(s)}{2}
    \bigl(
          (T_s - \id)P_tg
    \bigr)^2
    (x)
    \mright]^{\frac{1}{2}}
    \,\dd t
    \\
    &=
    \int_0^\infty
    e^{-\lambda t}
    \sqrt{\Gamma(P_tg)(x)}
    \,\dd t.
\end{align}
By Proposition~\ref{prop:rev-poincare-inequ}, we have $\Gamma(P_tg)(x) \leq \frac{1}{2t} \norm{g}_\infty^2$
and therefore
\begin{equation}
    \sqrt{\Gamma(f)(x)}
    \leq
    \frac{\norm{g}_\infty}{\sqrt{2}}
    \int_0^\infty
    e^{-\lambda t}t^{-1/2}\,\dd t.
\end{equation}
Substituting $u = \lambda t$ gives
\begin{equation}
    \int_0^\infty
    e^{-\lambda t}t^{-1/2}\,\dd t
    =
    \lambda^{-\frac12} \int_0^\infty
    e^{-u} u^{-1/2}\,\dd u
    = \sqrt{\frac{\pi}{\lambda}},
\end{equation}
where the integral is given by the Gamma function evaluated at $1/2$. Hence,
\begin{equation}
    \sqrt{\Gamma(f)(x)} \leq \sqrt{\frac{\pi}{2\lambda}} \, \norm{g}_\infty =
    \sqrt{\frac{\pi}{2\lambda}} \,  \norm{\bigl(
            \lambda\id+\Delta
        \bigr)f}_\infty.
\end{equation}
Squaring and taking the supremum over $x \in V$ yields
\begin{equation}
    \norm{\Gamma f}_\infty
    \leq
    \frac{\pi}{2\lambda}
    \norm{
        \lambda f + \Delta f
    }_\infty^2.
\end{equation}
\end{proof}

Using the previous lemmas, we can now prove Theorem~\ref{thm:main-shadow-bound}.

\begin{proof}[Proof of Theorem~\ref{thm:main-shadow-bound}]
\phantomsection
\label{proof:shadow-norm-bound}
Let $A$ be a $S$-sparse Hermitian operator. Decompose it into Pauli sectors
\begin{equation}
    A = \sum_{k \in \{0,1\}} \sum_{s \in \{0,1\}^n} A_{k,s},
    \quad
    A_{k,s} = \frac{1}{2^n} \sum_{\substack{b \in \ZZ_2^n \\ \chi_b(s) = (-1)^k}} \Tr(A P_{s,b}) P_{s,b},
\end{equation}
and recall from Proposition~\ref{prop:shadow-norm-short} that the shadow norm can be written as
\begin{equation}
    \norm{A}_{\mathrm{sh}}^2 = \sup_{\sigma \in \DDD} \Tr(\sigma \EE_{k \in \{0,1\}, s \in \{0,1\}^n} \bigl[ \MMM_{k,s}(\MMM^{-1}(A))^2 \bigr] ).
\end{equation}
Since $\sigma \in \DDD$ is a density matrix, it holds for any random positive-semidefinite operator $M$ that
\begin{equation}
    \Tr(\sigma \EE[M]) \leq \norm{\sigma}_1 \cdot \norm{\EE[M]}_\infty = \norm{\EE[M]}_\infty.
\end{equation}
Hence, after expanding the expectation, we obtain $\norm{A}_{\mathrm{sh}}^2 \leq \norm{T}_\infty$ with
\begin{equation}\label{eq:T-norm}
    T = \frac{1}{2} \sum_{k \in \{0, 1\}} \sum_{s \in \{0,1\}^n} p(s) \MMM_{k,s}(\MMM^{-1}(A))^2.
\end{equation}
Using Theorem~\ref{thm:main-channel-eigenvalues} and $A_{1,\zero} = 0$, we can write
\begin{equation}
    \MMM^{-1}(A) = \MMM^{-1}(A_{0,\zero}) + \sum_{k \in \{0, 1\}} \sum_{\substack{s \in \{0,1\}^n \setminus \{\zero\}}} \frac{2}{p(s)} A_{k,s}.
\end{equation}
Moreover, for $\ell \in \{0, 1\}$ and $t \in \{0,1\}^n \setminus \{\zero\}$, Proposition~\ref{prop:M-eigenvalues} and Lemma~\ref{lem:UPU-Z-type} imply
\begin{equation}
    \MMM_{k,s}(A_{\ell,t}) = \begin{cases}
        A_{k, s}, & \text{if $s = t$ and $k = \ell$}, \\
        0, & \text{otherwise}.
    \end{cases}
\end{equation}
Thus, for $s \neq \zero$, we have
\begin{equation}
    \MMM_{k,s}(\MMM^{-1}(A))
    = \MMM_{k,s}(\MMM^{-1}(A_{0,\zero})) + \frac{2}{p(s)} A_{k,s}.
\end{equation}
Substituting the previous equation into $T$, and using $(A + B)^2 \preceq 2A^2 + 2B^2$ for any Hermitian operators $A$ and $B$, yields
\begin{align}
    T
    \preceq \sum_{k \in \{0, 1\}} \sum_{s \in \{0,1\}^n} p(s) \MMM_{k,s}(\MMM^{-1}(A_{0,\zero}))^2 + \sum_{k \in \{0, 1\}} \sum_{s \in \{0,1\}^n}  \frac{4}{p(s)} A_{k,s}^2.
\end{align}
Since both terms are positive operators, we can apply $\norm{\cdot}_\infty$ and bound the first term using Lemma~\ref{lem:square-func-bound} to obtain
\begin{equation}\label{eq:T-two-parts}
    \norm{T}_\infty \leq \frac{2 + \pi}{p(\zero)} \cdot \norm{A_{0,\zero}}^2_\infty
    + \Big\lVert \sum_{k \in \{0, 1\}}\sum_{s \in \{0,1\}^n}  \frac{4}{p(s)} A_{k,s}^2 \Big\rVert_\infty.
\end{equation}
Note that $A_{0,\zero} = \diag(A)$ is precisely the diagonal part of $A$. If $\zero \notin S$, then $A_{0,\zero} = 0$. Otherwise,
\begin{equation}
    \frac{1}{p(\zero)} \leq \max_{s \in S} \frac{1}{p(s)},
    \qquad
    \norm{A_{0,\zero}}^2_\infty
    = \norm{\diag(A)}^2_\infty
    \leq
    \norm{A}^2_\infty.
\end{equation}
Thus, in both cases, we have
\begin{equation}\label{eq:diag-final-bound}
    \frac{2 + \pi}{p(\zero)} \cdot \norm{A_{0,\zero}}^2_\infty
    \leq
    (2 + \pi) \mleft[\max_{s \in S} \frac{1}{p(s)}\mright] \norm{A}_\infty^2.
\end{equation}
To bound the second term in Eq.~\eqref{eq:T-two-parts}, we use Lemma~\ref{lem:sector-square-sum} to obtain
\begin{equation}
    \sum_{k \in \{0,1\}}\sum_{s \in \{0,1\}^n} A_{k,s}^2
    \preceq
    \diag(A^2)
    \preceq
    \norm{A}^2_\infty \cdot \id.
\end{equation}
Using $A_{k,s}=0$ for $s \notin S$, it follows that
\begin{equation}\label{eq:off-diag-final-bound}
    \Big\lVert \sum_{k \in \{0, 1\}}\sum_{s \in \{0,1\}^n}  \frac{4}{p(s)} A_{k,s}^2 \Big\rVert_\infty
    \leq
    4 \mleft[\max_{s \in S} \frac{1}{p(s)}\mright] \norm{A}^2_\infty.
\end{equation}
By combining Eq.~\eqref{eq:diag-final-bound} and Eq.~\eqref{eq:off-diag-final-bound}, we obtain
\begin{equation}
    \norm{A}_{\mathrm{sh}}^2
    \leq
    \norm{T}_\infty
    \leq
    (6 + \pi) \mleft[\max_{s \in S} \frac{1}{p(s)}\mright] \norm{A}^2_\infty.
\end{equation}
\end{proof}

\section{Numerical experiments}\label{sec:experiments}
In this section, we present the numerical experiments used to assess the performance of the shadow protocols considered in this work. To this end, we implemented three GPU-compatible Python backends for sampling measurement outcomes from a quantum state $\rho$ according to a specified shadow protocol. The first backend relies on~\href{https://pennylane.ai/}{PennyLane}. The second is a custom Clifford simulator implemented in~\href{https://docs.jax.dev/en/latest/}{JAX}, based on the stabilizer-tableau algorithm of~\cite{aaronson2004improved}. The third is a custom JAX-based statevector simulator that supports batched execution of quantum circuits on a single GPU\@.

The experiments were performed on NVIDIA GH200 GPUs and consumed approximately 100 GPU-hours in total. The code is available on~\href{https://github.com/marcwannerchalmers/classical_shadows_SGM/}{github}.

\subsection{Exact \texorpdfstring{$(n-k)$}{(n-k)}-local \texorpdfstring{$X/Y$}{X/Y} observables}\label{sec:exp-XY}

We demonstrate the advantage of \NAME shadows on the following prediction task. Given $N$ samples obtained from an $n$-qubit state $\rho$, drawn at random from a prescribed family of quantum states, the goal is to estimate the expectation values of $M$ observables chosen uniformly at random from an $S$-sparse family of operators according to Definition~\ref{def:S-sparse}, with $|S|=\operatorname{poly}(n)$. As discussed previously in Section~\ref{sec:shadow-norm}, the advantage of our protocol over the Pauli or Clifford ensembles is expected to be most pronounced when $S$ corresponds to a ``non-example'' in the sense of~\cite{huang2020predicting}, such as a family of highly nonlocal Pauli observables whose relevant support set remains polynomially small in $n$.

\begin{figure}[t]
    \centering
    \includegraphics[width=0.8\textwidth]
    {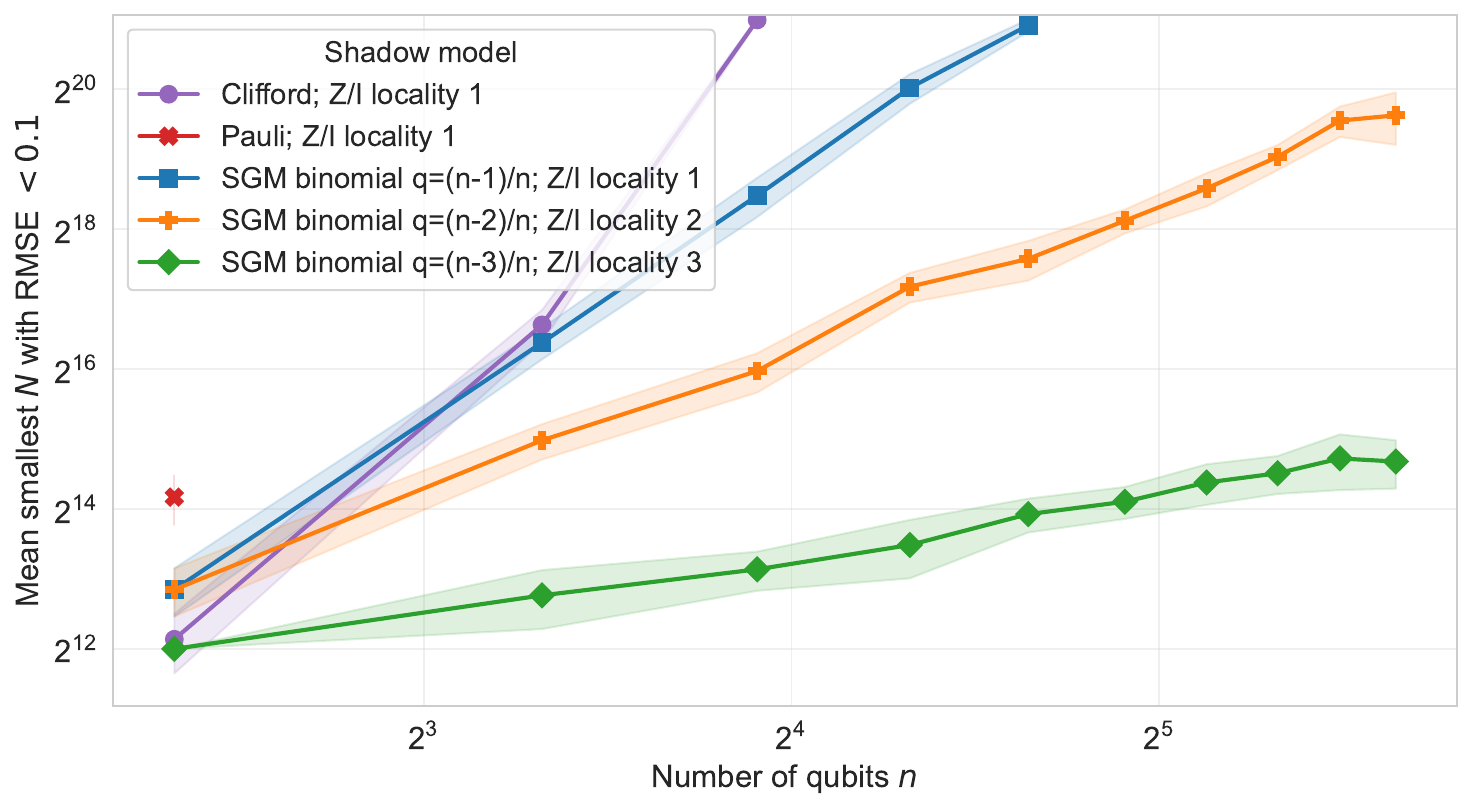}
    \caption{
        Sample complexity for estimating highly nonlocal observables on GHZ-type states.
        The number of shadow samples required to achieve
        $\mathrm{RMSE} \leq \varepsilon=0.1$ is shown as a function of the system size $n$.
        The observables are random Pauli operators that act as either $X$ or $Y$ on
        $n-k$ qubits, with $k=1,2,3$.
        For each value of $k$, \NAME is sampled from a binomial distribution with
        $q=(n-k)/n$, chosen to match the corresponding observable locality.
        Results are obtained from 50 randomly sampled observables and 10 independent
        repetitions for each system size.
    }
    \label{fig:ghz_sample_complexity}
\end{figure}

In the first experiment, we consider uniformly random instances of Pauli observables that act as either $X$ or $Y$ on all but $k$ qubits. The corresponding support masks are described by
\begin{equation}\label{exp:one_sparse}
    S_k
    =
    \mleft\{
        s \in \{0,1\}^n : \wt(s) = n-k
    \mright\},
    \qquad
    k=1,2,3.
\end{equation}
Thus, each observable has $X/Y$-support on exactly $n-k$ qubits, while the remaining $k$ qubits are padded by operators outside this support. We estimate the observables using the empirical mean to make the convergence behavior more transparent and to avoid ambiguities associated with tuning hyperparameters, such as the number of buckets in the median-of-means estimator.

Figure~\ref{fig:ghz_sample_complexity} shows the sample complexity for predicting random observables from the families defined above, with the underlying state drawn from the random GHZ-type ensemble described in Eq.~\eqref{eq:GHZ-type}. As a representative and widely studied class of highly entangled many-body states, this ensemble provides a natural benchmark for assessing the estimation of highly nonlocal observables. Sample points where the shadow prediction is trivial due to all snapshots being zero were removed. The numerical results show that the tuned \NAME protocol exhibits a sample complexity that grows polynomially with $n$, in agreement with Theorem~\ref{thm:main-shadow-bound}. This behavior contrasts with the substantially less favorable scaling obtained for Pauli and Clifford shadows on the same highly nonlocal prediction task.

\begin{figure}[t]
    \centering
    \begin{subfigure}[c]{0.49\textwidth}
        \centering
        \includegraphics[height=0.8\textwidth]
        {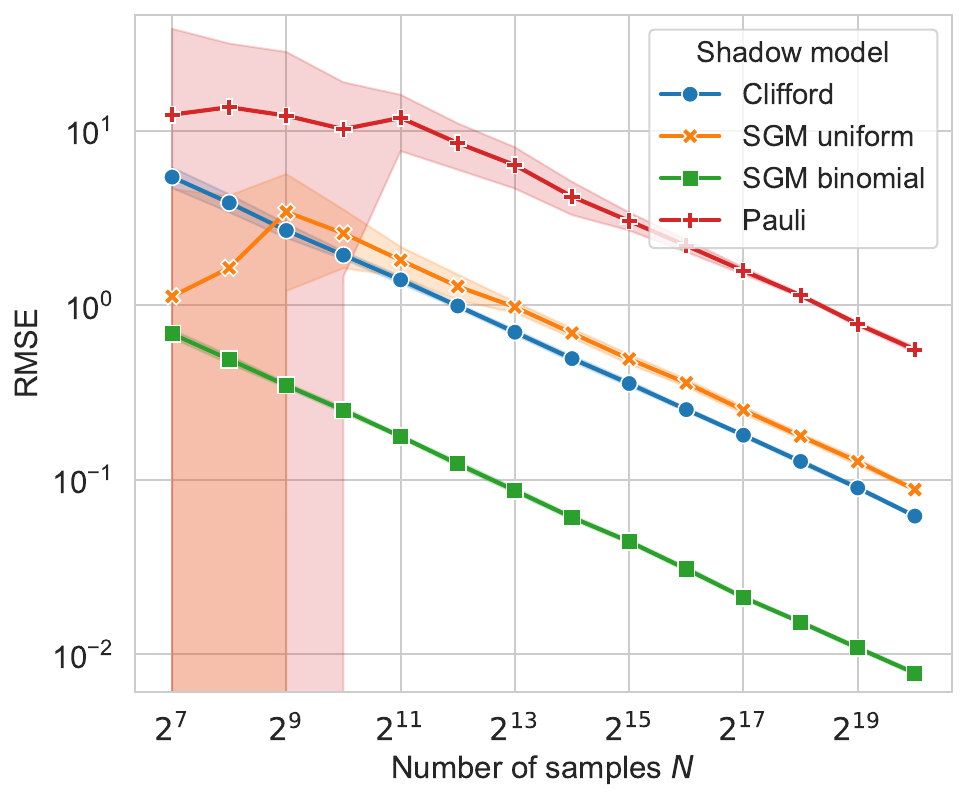}
        \caption{$n=12$}
        \label{fig:local_pauli_scaling_n12}
    \end{subfigure}
    \hfill
    \begin{subfigure}[c]{0.49\textwidth}
        \centering
        \includegraphics[height=0.8\textwidth]
        {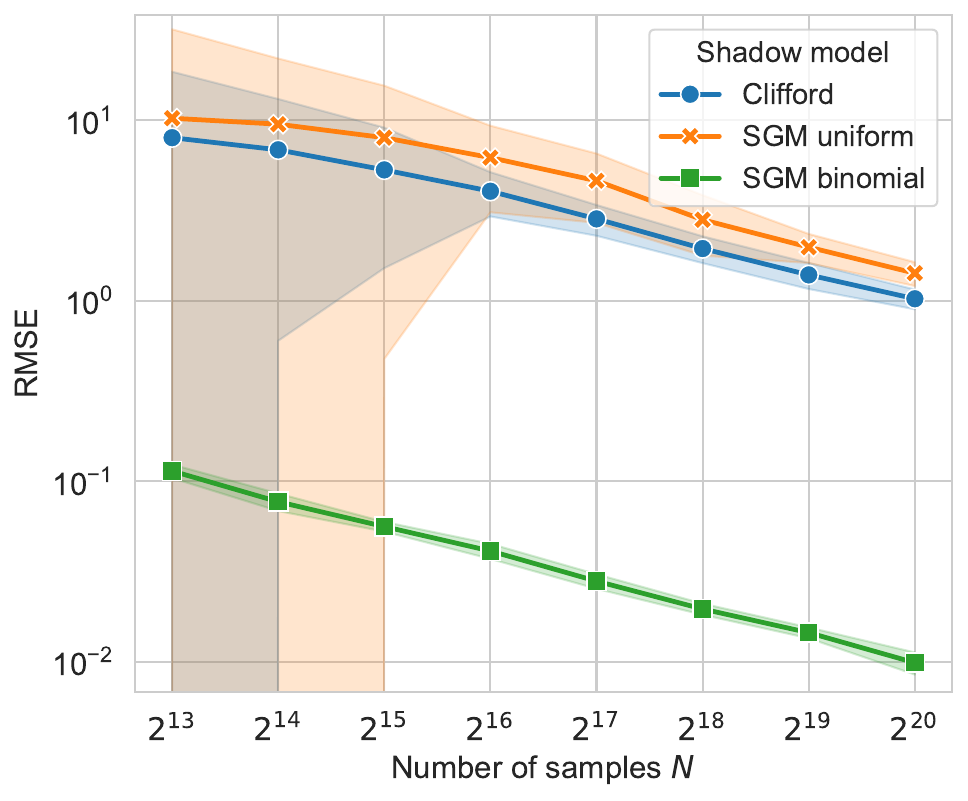}
        \caption{$n=20$}
        \label{fig:local_pauli_scaling_n20}
    \end{subfigure}

    \caption{
        Estimation error for highly nonlocal Pauli observables as a function of
        the number of shadow samples $N$, averaged over 10 independent Haar-random state
        realizations.
        We compare the Clifford and Pauli ensembles with the uniform \NAME ensemble and the binomial \NAME ensemble with $q=(n-1)/n$. The plotted quantity is the root-mean-square error (RMSE), and the shaded regions indicate variation across state instances.
    }
    \label{fig:xy_local_scaling}
\end{figure}

Figure~\ref{fig:xy_local_scaling} considers the analogous estimation problem for Haar-random states. Here, the observables act as either $X$ or $Y$ on all but one qubit. The results illustrate the poor suitability of the standard Pauli protocol for this task as well as the importance of adapting the shadow distribution to the target family of observables. In particular, increasing the system size from $n=12$ to $n=20$ leads to a significant decline in the performance of Clifford shadows and \NAME with a uniform sampling distribution. By contrast, the tuned \NAME protocol with $q=\frac{n-1}{n}$
exhibits a comparatively modest difference in prediction accuracy.

Two aspects of the RMSE curves are useful for interpreting the numerical results.
First, the RMSE for an observable drawn from the targeted ensemble is \emph{independent} of the number $M$ of observables used to estimate it. Since it relates to an empirical average rather than the worst-case prediction, it underlies the law of large numbers, such that increasing $M$ merely improves the statistical precision with which this quantity is estimated.
Second, for insufficiently large $N$, the empirical RMSE can substantially underestimate the asymptotic behavior predicted by Theorem~\ref{thm:main-shadow-bound} and may initially increase as $N$ grows. This finite-sample behavior reflects the sensitivity of the empirical mean to rare, large estimation errors and can be reduced by using a more robust estimator, such as median-of-means.

\subsection{\texorpdfstring{$k$}{k}-body reduced density matrix prediction}

In this section, we apply the \NAME ensemble to the physical setting studied in~\cite{zhao2021fermionic}, namely the prediction of elements of the $k$-body reduced density matrix of fermionic modes. For completeness, we introduce the necessary theory. Then we demonstrate how the distribution can be tailored to the respective family of observables, and finally compare the results of our numerical experiment with the method proposed in~\cite{zhao2021fermionic}.

\subsubsection{Fermionic modes}
A fermionic mode is a single one-particle degree of freedom that can be either unoccupied or occupied by one fermion, i.e.,\ a particle that follows the Pauli exclusion principle. Each mode $j$ has associated creation and annihilation operators $a_j$, $a_j^\dagger$. Many physical systems, especially molecules and closed electronic systems, have a well-defined number of fermions. Particle-number-preserving interactions are of particular interest because they respect this physical constraint.

The monomials $a_{j_1}^\dagger a_{j_1}, a_{j_1}^\dagger a_{j_2}^\dagger a_{j_1}a_{j_2},\dots$ form the basis for all such interactions. For a fixed-particle-number state $\rho$, the $k$-body reduced density matrix ($k$-RDM) is conventionally represented by the state obtained after tracing out all except $k$ modes, denoted by
\begin{equation}
{}^k D^{p_1\cdots p_k}_{q_1\cdots q_k}
=
\Tr\!\mleft(
a_{p_1}^\dagger\cdots a_{p_k}^\dagger
a_{q_k}\cdots a_{q_1}\rho
\mright).
\end{equation}
The monomials spanning fixed particle interactions can equivalently be expressed as Majorana monomials
\begin{equation}
\Gamma_A
=
(-\ii)^r\gamma_{\mu_1}\cdots\gamma_{\mu_{2r}},
\end{equation}
where the $2n$ Majorana operators
\begin{equation}
\gamma_{2j-1}=a_j+a_j^\dagger,
\qquad
\gamma_{2j}=-\ii(a_j-a_j^\dagger),
\qquad j=1,\ldots,n,
\end{equation}
exhibit the same algebraic properties as the Pauli operators. Therefore, there is a one-to-one correspondence between Majorana and Pauli operators, one of which is given by the Jordan--Wigner (JW) transformation,
\begin{equation}
\gamma_{2j-1}
=
Z_1\cdots Z_{j-1}X_j,
\qquad
\gamma_{2j}
=
Z_1\cdots Z_{j-1}Y_j.
\label{eq:JW}
\end{equation}
Hence every $\Gamma_A$ maps, up to a real sign, to an $n$-qubit Pauli operator.

Although the observables are local in the Majorana monomial basis, the Pauli operators obtained by the JW transformation are evidently nonlocal. In~\cite{zhao2021fermionic}, a shadow protocol based on Gaussian unitaries was introduced to address the shortcomings of the Pauli and Clifford protocols. However, for constant $k$, the \NAME protocol also provides an efficient alternative. Since observables corresponding to a $k$-RDM are spanned by Majorana monomials with degree at most $2k$, the respective Pauli observables given by the JW transformation must be $2k$-$X/Y$-local. Hence, Example~\ref{ex:binomial-shadow-norm} directly gives a sample complexity of $\OOO(n^{2k})$.

\subsubsection{Optimized \NAME distribution}
Although efficient, the aforementioned \NAME protocol does not exhaustively exploit the structure of the observables. By applying previous knowledge about the JW-transformed Majorana monomials, we propose two additional distributions for the \NAME protocol.
We start by stating the family of observables more rigorously. Physical fermionic operators lie in the space spanned by even-degree Majorana monomials, due to the parity superselection rule~\cite{streater2016pct}. Therefore, the target family is
\begin{equation}
    \mathcal O_{n,k}
    =
    \bigcup_{r=1}^{k}
    \mleft\{\Gamma_\mu:|\mu|=2r\mright\},
    \qquad
    L_{n,k}
    =
    |\mathcal O_{n,k}|
    =
    \sum_{r=1}^{k}\binom{2n}{2r}.
    \label{eq:rdm-target-family}
\end{equation}

By a slight abuse of notation, let $\lambda_{\Gamma_\mu}(p)$ denote the channel eigenvalue associated with the JW transformation of $\Gamma_\mu$. Then Theorem~\ref{thm:main-shadow-bound} relates its inverse to the corresponding shadow norm. Therefore, Proposition~\ref{prop:shadow-norm-sample-complexity} yields the following upper bound on the mean-squared error (MSE):
\begin{equation}\label{exp:ravg}
R_{\mathrm{avg}}(p)
\propto
\frac{1}{L_{n,k}}
\sum_{\Gamma_\mu\in\mathcal O_{n,k}}
\frac{1}{\lambda_{\Gamma_\mu}(p)}.
\end{equation}

To minimize this quantity, $p$ should be chosen such that $1/\lambda_{\Gamma_\mu}(p)$ is proportional to the corresponding multiplicity. The latter is the number of degree-$2r$ Majorana monomials with $X/Y$ support size $s$, which we denote by $N_{r,s}$ for $s=2,4,\ldots,2k$. Such a monomial contains $s$ singly occupied Majorana pairs, i.e.,\ pairs for which exactly one of $\mu_{2j-1}$ and $\mu_{2j}$ equals one. Each singly occupied mode admits two choices of Majorana, while the remaining $2r-s$ Majoranas form $r-s/2$ doubly occupied pairs chosen from the remaining $n-s$ modes, resulting in
\begin{equation}
    N_{r,s}
    =
    2^s\binom{n}{s}\binom{n-s}{r-s/2}
    \label{eq:rdm-degree-support-count}
\end{equation}
for even $s$ with invalid binomial coefficients interpreted as zero; $N_{r,s}=0$ for odd $s$. Write $N_s =\sum_{r=1}^{k}N_{r,s}$. Then
\begin{equation}
    N_0=\sum_{r=1}^{k}\binom{n}{r},
    \qquad
    N_s
    =
    2^s\binom{n}{s}
    \sum_{t=0}^{k-s/2}\binom{n-s}{t},
    \quad s=2,4,\ldots,2k.
    \label{eq:rdm-support-counts}
\end{equation}

When all occupied modes are doubly occupied, the JW transformation yields a $Z$-type observable. For a $Z$-support of size $t$, there are $\binom{n}{t}$ such observables. Substituting this count and the binomial distribution
\[
    p_q(a)=q^{\wt(a)}(1-q)^{n-\wt(a)}, \quad \hat{p}_q(b) = (1-2q)^{\wt(b)},
    \qquad 0<q<1,
\]
into Eq.~\eqref{exp:ravg} and applying Theorem~\ref{thm:main-channel-eigenvalues} gives
\begin{equation}
R_{\mathrm{avg}}(q)
=
\frac{1}{L_{n,k}}
\mleft[
\sum_{\substack{s=2\\ s \ \mathrm{even}}}^{2k}
\frac{2N_s}{q^s(1-q)^{n-s}}
+
\sum_{t=1}^{k}
\frac{2\binom{n}{t}}{1+(1-2q)^t}
\mright],
\label{eq:rdm-bernoulli-risk}
\end{equation}
which can be efficiently optimized over $q$ for fixed $n$ and $k$. However, the binomial distribution does not account for the fact that all target observables have even support.

A more flexible distribution first samples a mask weight
$w\in\{0,\ldots,n\}$ with probability $\pi_w$, and then selects
a mask uniformly among those of weight $w$:
\begin{equation}
    p_{\boldsymbol\pi}(a)
    =
    \frac{\pi_{\wt(a)}}{\binom{n}{\wt(a)}},
    \qquad
    \pi_w\geq0,
    \qquad
    \sum_{w=0}^{n}\pi_w=1.
    \label{eq:rdm-fixed-weight-distribution}
\end{equation}

For simplicity, we omit the $Z$-type observables, whose contribution can be uniformly upper-bounded. The resulting objective is
\begin{equation}
R_{\mathrm{avg}}^{\mathrm{JW}}(\boldsymbol\pi)
=
\frac{1}{L_{n,k}}
\sum_{\substack{s=2\\ \text{$s$ even}}}^{2k}
\frac{2N_s\binom{n}{s}}{\pi_s}.
\label{eq:rdm-fixed-weight-risk}
\end{equation}
Minimizing this expression over the probability simplex, subject to positive eigenvalues for all target observables, is a convex optimization problem with closed-form solution
\begin{equation}
\pi_s^{\mathrm{avg}}
=
\frac{\sqrt{N_s\binom{n}{s}}}
{\displaystyle
\sum_{\substack{u=2 \\ \text{$u$ even}}}^{2k}
\sqrt{N_u\binom{n}{u}}}.
\end{equation}
This follows from minimizing $\sum_s A_s/\pi_s$, for which the optimum satisfies $\pi_s\propto\sqrt{A_s}$. In practice, the optimizer may assign zero probability to masks that are irrelevant to the target family. To retain a globally invertible shadow channel, as required by Theorem~\ref{thm:main-shadow-bound}, we follow Example~\ref{ex:uniform-on-S} and assign a total probability mass $1-\delta>0$ uniformly to block encodings with $p_{\boldsymbol\pi}(a)=0$, where $1-\delta$ can be chosen below machine precision.

\subsubsection{Numerical experiment}
\begin{figure}[t]
    \centering
    \begin{subfigure}[t]{0.49\textwidth}
        \centering
        \includegraphics[height=0.8\textwidth]
        {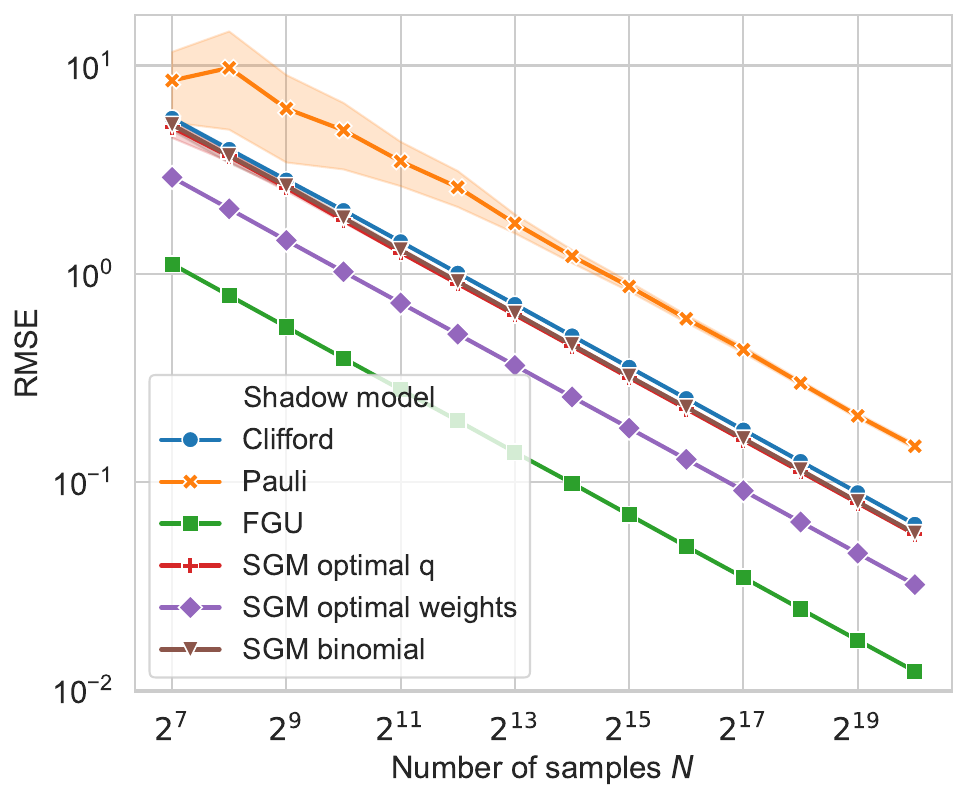}
        \caption{}
        \label{fig:majorana_rdm_scaling_rmse}
    \end{subfigure}
    \hfill
    \begin{subfigure}[t]{0.49\textwidth}
        \centering
        \includegraphics[height=0.8\textwidth]
        {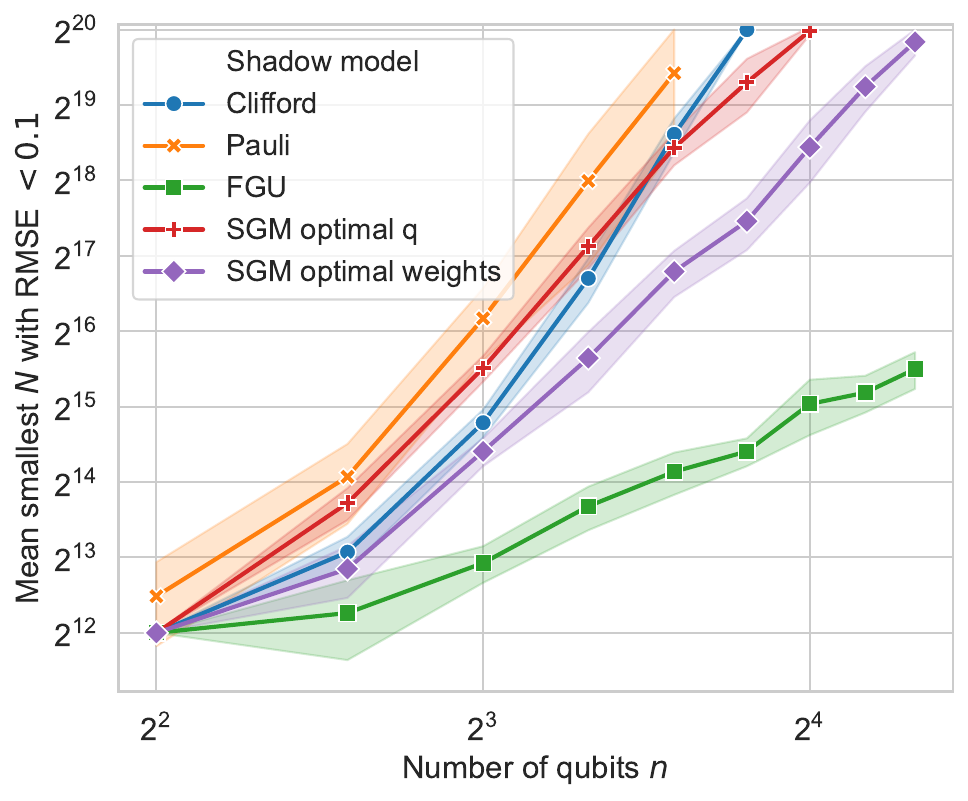}
        \caption{}
        \label{fig:majorana_rdm_scaling_sample_complexity}
    \end{subfigure}

    \caption{
        Performance for second-order Majorana RDM observables.
        \textbf{(a)} RMSE as a function of the number of shadow samples $N$
        for all 10,902 non-trivial observables of a 12-qubit Majorana state,
        averaged over 10 independent realizations.
        \textbf{(b)} Number of samples required to reach
        $\mathrm{RMSE}\leq 0.1$ as a function of system size $n$, using
        50 observables and 10 independent realizations per system size.
        We compare Clifford, Pauli, and fermionic Gaussian (FGU) shadows with
        \NAME using optimized binomial and fixed-weight sampling.
    }
    \label{fig:majorana_rdm_scaling}
\end{figure}

In our numerical experiment, we simulate $k$-RDM estimation for random states with particle number fixed to $n/2$. These states are obtained by projecting a Haar-random state onto the $n/2$-particle subspace and renormalizing. Figure~\ref{fig:majorana_rdm_scaling} shows that the \NAME protocol outperforms both the Pauli and Clifford protocols in this setting, with additional practical improvements obtained by tailoring the sampling distribution. Nevertheless, as suggested by Eq.~\eqref{eq:rdm-fixed-weight-distribution} with $\pi_s^{\mathrm{avg}}$, the MSE retains its $\OOO(n^{2k})$ scaling.

The measurement protocols also differ in their circuit complexity.
The fermionic Gaussian unitaries of~\cite{zhao2021fermionic} are implemented using random fermionic SWAP gates, which generally require $\OOO(n)$ Pauli rotation gates and yield circuits of depth $\OOO(n)$ with $\OOO(n^2)$ gates. While the practical implementation can be improved, for example using Majorana SWAP networks~\cite{fisher2026improving}, to the best of our knowledge the asymptotic gate count remains unchanged. In contrast, the \NAME protocol corresponds to a Moore-Nilsson circuit~\cite{moore2001parallel}, requiring depth $\OOO(\log(n))$ and $\OOO(n)$ gates in general. For the present example, its circuit complexity is constant for fixed $k$, since $\wt(a)\leq k$. Similar constructions can be applied, for example, to estimating the energy of the Kitaev chain~\cite{vodola2014kitaev}.

\section{Discussion}\label{sec:discussion}

In this work, we introduced selective GHZ measurement ensembles as a tunable family of classical shadow protocols. We determined the complete spectrum of the shadow channel for arbitrary probability distributions and derived a shadow-norm bound for structured observables with prescribed $X/Y$ support. These results place \NAME shadows between local measurement schemes, such as Pauli shadows, and nonlocal schemes, such as global Clifford shadows. While we do not strictly outperform either one, the underlying probability distribution can be adapted to certain families of target observables, leading to smaller worst-case shadow norms for these families while retaining comparatively simple measurement circuits. Notably, these families can contain exponentially many non-identity Pauli observables. Our numerical results support the theoretical predictions and further illustrate the flexibility of the protocol in a physically motivated setting, namely the estimation of observables in fermionic many-body systems.

In the following, we present an overview of related work on classical shadows to place our results in context and highlight connections to prior work. We then discuss limitations of our approach and outline several open questions and directions for future research.

\subsection{Related work}\label{sec:related-work}

Classical shadows provide a general framework for predicting many properties of an unknown quantum state from relatively few measurements~\cite{huang2020predicting,huang2022learning}. Standard constructions based on local Pauli and global random Clifford measurements occupy opposite ends of a \textit{circuit} versus \textit{sample} complexity tradeoff: local Pauli measurements are experimentally simple and efficient for local observables but become costly for operators with large support, whereas global Clifford measurements can efficiently estimate low-rank observables at the expense of substantially greater circuit depth and gate overhead.

This tradeoff has motivated the development of intermediate measurement ensembles that retain some of the statistical advantages of global scrambling while remaining experimentally accessible. Finite-depth local Clifford circuits provide one such route, interpolating naturally between local Pauli and global Clifford measurements. Tensor-network methods can be used to characterize the resulting shadow channel and reconstruction map~\cite{bertoni2024shallow,akhtar2023scalable}, and these protocols can yield improvements over Pauli measurements for quasi-local observables while requiring much shallower circuits. More broadly, these results highlight the central role of controlled operator spreading: increasing circuit depth can improve measurement efficiency by spreading information across the system, but only up to the point where the additional scrambling no longer compensates for the increased implementation cost.

This perspective naturally extends from engineered random circuits to native many-body dynamics, where the physical evolution itself supplies the required scrambling. Hamiltonian-shadow protocols exploit chaotic or locally scrambled dynamics to generate informative measurement ensembles~\cite{hu2022hamiltonian,hu2023classical}, while many-body-localized dynamics provides a complementary regime in which information spreads only quasi-locally~\cite{zhou2024efficient}. Related approaches use ancillas and emergent state designs generated by global evolution~\cite{mcginley2023shadow}, as well as fixed-quench analog protocols in enlarged Hilbert spaces~\cite{tran2023measuring}. More recent work has shown that even a single Hamiltonian, sampled at randomized evolution times, can suffice for general state-property prediction~\cite{liu2024predicting}, with extensions to fermionic systems~\cite{denzler2024learning} and to chaos-assisted conversion of classical randomness into quantum randomness for shadow tomography~\cite{mok2025optimal}.

The challenge of characterizing and efficiently inverting the shadow channel has also led to connections with several areas of mathematics. Representation-theoretic ideas already appear in robust shadow estimation, where symmetry and Schur-type decompositions are used to simplify the action of the shadow channel~\cite{chen2021robust}. Similar structure arises in fermionic settings, where Gaussian and matchgate ensembles decompose operator space into invariant sectors, enabling explicit reconstruction maps and efficient estimation of fermionic observables~\cite{zhao2021fermionic,wan2023matchgate}. These examples motivated increasingly general formulations of shadow tomography: frame-theoretic approaches describe reconstruction for arbitrary measurement schemes through dual frames~\cite{innocenti2023shadow}, while Pauli-invariant ensembles exploit diagonalization of the shadow channel in the Pauli basis~\cite{bu2024classical}. Group-theoretic formulations further organize the channel according to irreducible representation sectors and connect shadow reconstruction with symmetry-based error mitigation~\cite{zhao2024group-theoretic}. More recently, this perspective has been extended beyond group ensembles to random measurements associated with symmetric spaces, further broadening the class of structured shadow channels for which analytical reconstruction can be developed~\cite{chang2026classical}.

Experimental demonstrations have established the practicality of classical shadow protocols across several quantum platforms~\cite{zhang2021experimental,struchalin2021experimental}. In parallel, substantial effort has focused on robustness to experimental imperfections. Calibration-based approaches introduced robust shadow estimation for noisy randomized measurements~\cite{chen2021robust}, followed by general analyses of classical shadows under noisy quantum channels~\cite{koh2022classical}. Classical shadow data has also been incorporated directly into quantum-error-mitigation strategies, including shadow distillation and error-mitigated classical shadows~\cite{seif2023shadow,jnane2024quantum}, while symmetry and group structure have been exploited to suppress errors in shadow estimation~\cite{zhao2024group-theoretic}. More recent work has investigated how noise modifies the optimal amount of scrambling~\cite{rozon2024optimal}, established stability guarantees under gate-dependent noise~\cite{brieger2025stability}, and developed robust protocols for shallow measurement circuits~\cite{farias2025robust,hu2025demonstration}. These advances have increasingly connected noise mitigation with experimentally realizable shadow protocols, including direct demonstrations of error-mitigated and shallow-shadow techniques on quantum hardware~\cite{hu2025demonstration}.

Beyond reducing the measurement cost of quantum-state characterization, classical shadows have developed into a broadly useful primitive for extracting and reusing information from quantum experiments. The original framework demonstrated that a single randomized-measurement dataset could be used to estimate many observables, fidelities, correlation functions, and entanglement-related quantities simultaneously~\cite{huang2020predicting}. This idea was quickly extended to nonlinear properties of quantum states, including mixed-state entanglement and moments of partially transposed density matrices~\cite{elben2020mixed}, as well as quantum Fisher information and multipartite-entanglement certification~\cite{rath2021quantum,vitale2024robust}. Classical shadows have also become an important tool for quantum chemistry and variational algorithms, where the simultaneous estimation of many Hamiltonian terms can substantially reduce measurement overhead; this has motivated derandomized and locally biased measurement strategies tailored to molecular Hamiltonians~\cite{huang2021efficient,hadfield2022measurements}. Structure-adapted variants further enable efficient estimation of fermionic reduced density matrices and other observables relevant to electronic-structure calculations~\cite{zhao2021fermionic,wan2023matchgate}. Beyond observable estimation, shadows have been incorporated directly into variational optimization, for example to diagnose and avoid barren plateaus or to train variational circuits using large collections of covariance constraints~\cite{sack2022avoiding,boyd2022training}, and more recently to construct large quantum-subspace expansions using only classical post-processing of previously acquired shadow data~\cite{boyd2025high}. Another major direction applies shadows to dynamical and process characterization: randomized measurements have been used to probe scrambling and out-of-time-order correlations~\cite{garcia2021quantum, mcginley2022quantifying}, while shadow process tomography extends the framework from states to quantum channels and enables applications such as process characterization and Hamiltonian learning~\cite{kunjummen2023shadow,levy2024classical}.

The same reusable-data perspective has found applications in finite-temperature physics through pure thermal shadows~\cite{coopmans2023predicting}, in cross-platform verification of independently prepared quantum states~\cite{zhu2022cross}, and in spectroscopy, where time-resolved shadow data can be combined across many observables to infer spectral information and energy gaps~\cite{chan2025algorithmic}. Classical shadows have additionally become a useful interface between quantum experiments and classical machine learning, allowing experimentally generated shadow data to be used for phase classification and prediction of many-body properties~\cite{cho2024machine}, and enabling quantum-trained models whose inference can be subsequently performed classically~\cite{jerbi2024shadows}. Taken together, these developments have shifted the role of classical shadows from a sample-efficient tomography protocol to a general-purpose quantum-to-classical data representation that supports downstream tasks across quantum chemistry, many-body physics, metrology, spectroscopy, process learning, error mitigation, and machine learning.

\subsection{Limitations and outlook}\label{sec:outlook}

On the theoretical side, our work leaves several open questions.
While we derived the complete spectrum of the shadow channel and showed that the \NAME ensemble is tomographically complete for strictly positive distributions, this is only a sufficient condition. It remains an open question whether allowing
$p(\zero)=0$ can lead to smaller sample-complexity guarantees for suitable
classes of observables. Another open question is whether the \NAME shadow-norm bound in Theorem~\ref{thm:main-shadow-bound} is sharp. While the bound is attained for single Pauli observables up to a constant factor, it remains open whether there exist classes of observables for which it can be significantly improved. It also remains open how much the constant $6+\pi$ in the proof of Theorem~\ref{thm:main-shadow-bound} can be reduced.

On the practical side, we have mostly focused on the sample complexity of classical shadows, but other resources are also relevant. In particular, the circuit depth and connectivity required to implement the measurement ensemble can be limiting factors in near-term experiments. Although \NAME is designed to be experimentally accessible, it remains an open question how its implementation cost compares to other shadow protocols in practice, especially in the presence of noise and limited connectivity.
Additionally, it would be interesting to find more classes of practically relevant observables for which \NAME could provide an advantage over other shadow protocols
and to perform more extensive numerical studies to benchmark its performance in realistic scenarios.

Finally, there are multiple directions to extend our work. One natural direction is to incorporate a noise model into the \NAME ensemble and analyze its robustness to experimental imperfections. In this case, it may still be possible to compute the exact spectrum of the shadow channel and derive explicit shadow-norm bounds for various noise models. Another direction would be to generalize our framework to other types of entangled measurements. For example, one could consider measurements based on graph states or other stabilizer states, which may provide advantages for different classes of observables.

\section*{Acknowledgments}

NF, DD, AFK, and MW are currently supported by SSF (Swedish Foundation for Strategic Research), grant number FUS21-0063.
Computational resources were provided by the National Academic Infrastructure for Supercomputing in Sweden (NAISS), funded by the Swedish Research Council.
PPN is supported by the project “Quantum computing for future mobility solutions” funded by Chalmers.
AG and AFK are supported by the Knut and Alice Wallenberg Foundation through the Wallenberg Centre for Quantum Technology (WACQT). AFK also acknowledges support from the Swedish Foundation for Strategic Research (grant number FFL21-0279), the Horizon Europe programme HORIZON-CL4-2022-QUANTUM-01-SGA via the project 101113946 OpenSuperQPlus100, and from the Norwegian Research Council through the Norwegian Quantum Software Center (NorQSoft, project number 361350).

\paragraph{Declaration of generative AI use.}
The main ideas of the work were developed by us, and we worked out and wrote down all technical details.
ChatGPT 5.6 Sol was used for searching the literature, proofreading, improving wording and developing proof ideas. Notably, it helped us to discover the graph
construction and the connection to Markov semigroups in the proof of Theorem~\ref{thm:main-shadow-bound}. Furthermore, it assisted in implementing and debugging parts, as well as polishing the final version of the codebase.
All AI-generated material was reviewed and edited by us, and we take full responsibility for the manuscript. Prior work has been credited to the best of our knowledge.

\bibliographystyle{alphaurl}
\bibliography{references}

\end{document}